\documentclass[
  accepted=2023-05-17,
  a4paper,twocolumn
]{quantumarticle}
\pdfoutput=1

\usepackage{stylesheet}

\begin{document}

\notoc

\title{Page curves and typical entanglement in linear optics}
\date{May~17,~2023}

\author{Joseph~T.~Iosue}\orcid{0000-0003-3383-1946}
\email{jtiosue@umd.edu}
\affiliation{\JQI}
\affiliation{\QUICS}

\author{Adam~Ehrenberg}\orcid{0000-0002-3167-6519}
\affiliation{\JQI}
\affiliation{\QUICS}

\author{Dominik~Hangleiter}\orcid{0000-0002-4766-7967}
\affiliation{\QUICS}
\affiliation{\JQI}

\author{Abhinav~Deshpande}\orcid{0000-0002-6114-1830}
\affiliation{Institute~for~Quantum~Information~and~Matter,~California~Institute~of~Technology,~Pasadena,~CA~91125,~USA}

\author{Alexey~V.~Gorshkov}\orcid{0000-0003-0509-3421}
\affiliation{\JQI}
\affiliation{\QUICS}

\begin{abstract}
Bosonic Gaussian states are a special class of quantum states in an infinite dimensional Hilbert space that are relevant to universal continuous-variable quantum computation as well as to near-term quantum sampling tasks such as Gaussian Boson Sampling. 
In this work, we study entanglement within a set of squeezed modes that have been evolved by a random linear optical unitary. 
We first derive formulas that are asymptotically exact in the number of modes for the R\'enyi-2 Page curve (the average R\'enyi-2 entropy of a subsystem of a pure bosonic Gaussian state) and the corresponding Page correction (the average information of the subsystem) in certain squeezing regimes. 
We then prove various results on the typicality of entanglement as measured by the R\'enyi-2 entropy by studying its variance. 
Using the aforementioned results for the R\'enyi-2 entropy, we upper and lower bound the von Neumann entropy Page curve and prove certain regimes of entanglement typicality as measured by the von Neumann entropy. 
Our main proofs make use of a symmetry property obeyed by the average and the variance of the entropy that dramatically simplifies the averaging over unitaries. 
In this light, we propose future research directions where this symmetry might also be exploited.
We conclude by discussing potential applications of our results and their generalizations to Gaussian Boson Sampling and to illuminating the relationship between entanglement and computational complexity.
\end{abstract}

\maketitle

\section{Introduction}
\label{sec:intro}

In his pioneering paper \cite{page_average_1993}, Page considered a Haar-random pure state in an $N$-dimensional Hilbert space. He conjectured an exact formula for the average entanglement entropy of the reduced density matrix on an $rN$-dimensional subspace, where $r \in [0,1]$. His conjecture was proven soon after \cite{foong_proof_1994,sanchez-ruiz_simple_1995,sen_average_1996}. This average, which is itself a function of $r$ and $N$, is now known as the \textit{Page curve}. Page also defined the \textit{average information of the subsystem} as the difference between the maximum of the entropy and the average entropy. This quantity has since come to be known as the \textit{Page correction}. Since then, the Page curve and correction have found applications in a variety of areas, such as black holes \cite{page_information_1993,hayden_black_2007,bianchi_entanglement_2015}, quantum information theory \cite{hayden_aspects_2006,hosur_chaos_2016}, and statistical mechanics \cite{fujita_page_2018,lu_renyi_2019,nakagawa_universality_2018,vidmar_entanglement_2017,vidmar_entanglement_2017-1,hackl_average_2019,goldstein_canonical_2006,dalessio_quantum_2016}, among others.
As a next step, many works considered the typical deviation of the entanglement entropy from its average. In particular, a system is said to exhibit \textit{typical entanglement} if there is a vanishingly small probability that a random state has entanglement bounded away from the average. This phenomenon was introduced and studied in a variety of systems \cite{hayden_randomizing_2004,collins_matrix_2013,hastings_random_2015,garnerone_typicality_2010,popescu_entanglement_2006,gross_most_2009,bremnerAreRandomPure2009,bianchi_typical_2019,bianchi_page_2021,dahlsten_entanglement_2014}.

Entanglement is a key feature of quantum physics and can be used as a resource to complete various tasks, such as teleportation, key distribution, dense coding, and many others \cite{nielsen_quantum_2010,bengtssonGeometryQuantumStates2008,amicoEntanglementManybodySystems2008,horodeckiQuantumEntanglement2009,wilde_quantum_2017}.
Furthermore, entanglement is a necessary ingredient for quantum advantage, since quantum computations with little entanglement can be simulated efficiently classically \cite{jozsaRoleEntanglementQuantumcomputational2003,vidal_efficient_2003,verstraete_matrix_2004}.
One can define quantum advantage using the language of complexity as using quantum resources to perform a task that is classically hard but quantumly easy. One such task is sampling from the output probability distribution of a Gaussian Boson Sampling experiment \cite{lundBosonSamplingGaussian2014,hamilton_gaussian_2017,kruseDetailedStudyGaussian2019,deshpande_quantum_2022,grierComplexityBipartiteGaussian2022}.  
The general relationship between entanglement and complexity is largely unknown, but at least some entanglement is necessary for classical hardness. 
Similarly, in the setting of Boson Sampling, at least some amount of non-Gaussianity is necessary for sampling hardness \cite{chabaudResourcesBosonicQuantum2022}.
On the other hand, too much entanglement can cause a state to be useless for computation. In particular, the typical (over the Haar measure) finite-dimensional quantum state is too entangled to be useful for computation \cite{gross_most_2009,bremnerAreRandomPure2009}. Thus, studying average and typical entanglement is necessary for learning about the useful part of entanglement and what utility random states have.

Entanglement in infinite-dimensional quantum states generated by bosonic Gaussian inputs and Gaussian operations has found direct application in areas such as quantum sensing \cite{zhuangDistributedQuantumSensing2018,polinoPhotonicQuantumMetrology2020,ohOptimalDistributedQuantum2020,malitestaDistributedQuantumSensing2021,barbieriOpticalQuantumMetrology2022} and quantum communication \cite[Ch.~12]{adessoEntanglementGaussianStates2007}. Furthermore, the original motivation for studying the finite-dimensional Page curve was to study the black hole information paradox \cite{page_average_1993,page_information_1993}. Since the degrees of freedom around black holes are inherently infinite-dimensional bosonic (e.g.~photonic) modes, an infinite-dimensional bosonic Page curve may help to understand black hole information dynamics \cite[Ch.~14]{adessoEntanglementGaussianStates2007}.

\textit{Our contributions.}
In this work, we study Page curves and the typicality of entanglement in continuous-variable bosonic Gaussian states. Specifically, we compute average and typical entanglement quantities averaged over passive (energy-conserving) Gaussian unitaries, also called linear optical unitaries, with a fixed initial product state of squeezed vacuum states on $n$ modes.
Indeed, this setup is exactly that of a Gaussian Boson Sampling experiment, which is of great recent interest due to experimental claims of quantum advantage via Gaussian Boson Sampling \cite{zhong_quantum_2020,zhongPhaseProgrammableGaussianBoson2021,madsen_quantum_2022}.

We describe this setup in more detail in \cref{sec:setup}. Then in \cref{sec:expval}, we begin by studying the regime when all the initial squeezing strengths are equal, and denote this squeezing strength by $s$. We derive an analytic expression for the average R\'enyi-2 entropy of a subsystem of $k$ modes as a function of $s$, $n$, and $k$ that is exact asymptotically in $n$ for arbitrary values of $s$ and $k$. Using this expression, we exactly compute the corresponding Page correction. These results are summarized in \cref{fig:page-curve}. 
In \cref{sec:variance}, we then study the presence of typical entanglement for various scalings of $k$ with $n$. 
When the distance between the entanglement of a random state and the average entanglement value vanishes additively (resp.~multiplicatively), we say that entanglement is strongly (resp.~weakly) typical.
We prove that entanglement as measured by the R\'enyi-2 entropy is weakly typical for any $k$ and strongly typical whenever $k\in\littleo{n}$. We further show that entanglement as measured by the von Neumann entropy is weakly typical whenever $k \in \littleo{n}$. Finally, in \cref{sec:unequal-squeezing}, we generalize our discussion to the regime when the initial squeezing strengths are not necessarily equal. We show that, if a certain conjecture is true, then the R\'enyi-2 and von Neumann entanglement entropies are both weakly typical whenever $k\in \littleo{n}$.

\textit{Prior work.}
Refs.~\cite{bianchi_page_2021,bhattacharjeeEigenstateCapacityPage2021} studied Page curves and entanglement in fermionic Gaussian states.
Serafini et al.\ have considered typical entanglement in bosonic Gaussian states \cite{serafini_canonical_2007,serafini_teleportation_2007}, where they defined two measures, namely the microcanonical and canonical measures, on the set of all $n$-mode bosonic Gaussian states and averaged with respect to these measures. Roughly, averaging over the microcanoncial measure corresponds to integrating over all bosonic Gaussian states up to a fixed bounded total energy, and averaging over the canonical measure corresponds to integrating over all bosonic Gaussian states with a Boltzmann weight factor decaying with the energy of the state.
Fukuda and Koenig generalized these results by studying entanglement averaged over passive Gaussian unitaries with a fixed initial product state of squeezed vacuum states on $n$ modes \cite{fukuda_typical_2019}. Indeed, their setup is exactly the one we consider in this work. As noted in Ref.~\cite{fukuda_typical_2019}, the measure defined by fixing squeezing strengths and then applying a random passive Gaussian unitary generalizes both the microcanonical and canonical measures, as the latter two measures can be expressed as convolutions of the former with certain distributions on the set of all squeezing configurations.

Serafini et al.~and Fukuda and Koenig consider entanglement in a subsystem of $k$ bosonic modes when $k \in \littleo{n}$; roughly, Serafini et al.~allow $k \in \bigO{1}$ and Fukuda and Koenig allow $k \in \littleo{n^{1/3}}$. To the best of our knowledge, there are currently no results on average entanglement or typical entanglement in the regime when $k \in \bigOmega{n^{1/3}}$. This is what we address in our work.
Concerning typical entanglement, we emphasize that our results do not supplant the results of Ref.~\cite{fukuda_typical_2019} in general, but rather only in certain regimes. In particular, we primarily consider the situation of equal squeezing strengths, whereas their results pertain to the situation of arbitrary squeezing. Similarly, we primarily consider the R\'enyi-2 entropy, whereas their results apply to both the R\'enyi-2 and von Neumann entropies. We summarize our work and that of Ref.~\cite{fukuda_typical_2019} on typical entanglement in bosonic Gaussian states in \cref{tab:typicality}.

For general quantum states, the von Neumann entropy has certain properties, such as \textit{strong subadditivity}, that the R\'enyi-$\alpha$ entropies do not. For other such properties, we refer to Ref.~\cite{wilde_quantum_2017}. Because of this, the von Neumann entropy is generally considered a better measure of entanglement than the R\'enyi-$\alpha$ entropies. 
Notably, however, it has been shown that the R\'enyi-2 entropy is special when restricting to bosonic Gaussian states. For example, it was recently proven that for bosonic Gaussian states, the R\'enyi-2 entropy also obeys strong subadditivity \cite{adesso_measuring_2012,camilo_strong_2019}. 
The R\'enyi-2 entropy is also equal, up to a constant, to the phase-space Shannon sampling entropy $-\int W(\bm x,\bm p)\log W(\bm x,\bm p)\Dd{n}{\bm x}\Dd{n}{\bm p}$ of the Wigner distribution $W(x,p)$ of the state for Gaussian states \cite{adesso_measuring_2012}. Here the position vector $\bm x$ and momentum vector $\bm p$ parameterize the phase space of $n$ oscillator modes. The phase-space Shannon sampling entropy has an operational meaning in terms of sampling via homodyne detections \cite{adesso_measuring_2012,buzek_sampling_1995}.
Furthermore, it has been shown that for pure tripartite Gaussian states, the R\'enyi-2 entropy obeys a strong subadditivity inequality that is stronger than that for the von Neumann entropy \cite{adesso_strong_2016}. 
Finally, correlation measures for Gaussian states based on the R\'enyi-2 entropy have been found that have no counterpart when using the von Neumann entropy \cite{lami_schur_2016}. 
As noted in Ref.~\cite{adesso_measuring_2012}, the aforementioned results are ``planting the seeds for a full Gaussian quantum information theory based on the R\'enyi-2 entropy.'' Our work is focused primarily on entanglement in bosonic Gaussian states as measured by the R\'enyi-2 entropy and can thus be viewed as planting a few more seeds.

\section{Setup}
\label{sec:setup}

In this work, we consider a linear optical system of $n$ modes, where each mode is an independent quantum harmonic oscillator. We will restrict our attention to bosonic Gaussian states. We refer to Refs.~\cite{serafini_quantum_2017,fukuda_typical_2019} for background on the theory of Gaussian states, and we provide the necessary details to understand this paper in \cref{ap:gaussian-states}.

Consider an $n$-mode mixed state $\rho$ in the Hilbert space of square integrable wavefunctions $\calH = L^2(\bbR)^{\otimes n}$. For each $i \in \set{1,\dots, n}$, let $\hat x_i$ and $\hat p_i$ be the position and momentum quadrature operators on the $i^{\rm th}$ mode. For each $i$, define $\hat r_{i} \coloneqq \hat x_i$ and $\hat r_{n+i} \coloneqq \hat p_i$. $\rho$ is a Gaussian state if there is a $\beta > 0$ and a Hamiltonian $\hat H$ that is at most quadratic in the quadrature operators such that $\rho$ is the thermal state $\rho \propto \e^{-\beta \hat H}$. Since $\hat H$ contains only linear and quadratic quadrature terms in $\hat r_i$, $\rho$ is fully characterized by its first and second moments, $\Tr(\rho \hat r_i)$ and $\sigma_{ij} = \frac{1}{2}\Tr[\rho (\hat r_i \hat r_j + \hat r_j \hat r_i)] - \Tr(\rho \hat r_i)\Tr(\rho \hat r_j)$. $\sigma$ is called the \textit{covariance matrix} of the state $\rho$.

A unitary $U$ is Gaussian if there exists a Hamiltonian $\hat H$ that is at most quadratic in the quadrature operators such that $U = \e^{-\i \hat H}$. One can show that a Gaussian unitary maps Gaussian states to Gaussian states. Any Gaussian state can be generated by acting on the vacuum state with a Gaussian unitary.  The set of all Gaussian unitaries is isomorphic to the real symplectic group of $2n\times 2n$ matrices $\Sp(2n)$. By the Euler decomposition theorem of a symplectic matrix, it follows that any pure Gaussian state can be generated by acting on an initial product state of squeezed vacuum modes with a passive (energy conserving) Gaussian unitary. The set of all passive Gaussian unitaries is isomorphic to $\Sp(2n) \cap \O(2n) \cong \U(n)$, where $\O(2n)$ is the orthogonal group of $2n\times 2n$ matrices and $\U(n)$ is the unitary group of $n\times n$ matrices. $\Sp(2n)$ is not compact and thus does not have a finite Haar measure to average over. Physically, this is due to the fact that the Gaussian operation of squeezing can take on unbounded values in $\bbR$. However, $\U(n)$ is compact, and it is hence well-defined to consider uniformly sampling from $\U(n)$ according to the finite Haar measure. 

We are therefore motivated to consider the following notion of a random pure Gaussian state. Namely, we initialize the $i^{\rm th}$ mode to be in a squeezed vacuum state with fixed squeezing parameter $s_i \in \bbR$ for each $i \in \set{1,\dots, n}$. We then randomly sample a passive Gaussian unitary from $\U(n)$ and apply it to the $n$ modes. This random state is thus characterized by a fixed choice of the squeezing parameters $s_i$ and a Haar-random choice of a passive Gaussian unitary $U$. Understanding the properties of such random states is of great interest, particularly because Gaussian Boson Sampling experiments rely on precisely those state preparations. 
For squeezing parameters $s_i$ for $i\in\set{1,\dots,n}$, the total expected number of bosons of the state on the $n$ modes is $\sum_{i=1}^n \sinh^2(s_i)$. For simplicity, we will begin by considering the case when all the squeezing parameters are equal; $s_i = s$ for each $i$ for some fixed $s$. In the case of equal squeezings, the average total boson number per mode is $\sinh^2 s$. The general case of unequal squeezings will be discussed in \cref{sec:unequal-squeezing}.

The $n$ modes are then partitioned into two groups --- one group of $k = r n$ modes for some $0\leq r \leq 1$, and one group of $n-k = (1-r)n$ modes. We then compute the entropy of the reduced state of the $k$ modes, or equivalently, since we are considering pure states, the entropy of the reduced state of the $n-k$ modes \cite{nielsen_quantum_2010}. Let the density matrix of the reduced state be $\rho$. For the entropy function, we will be primarily focused on the R\'enyi-2 entropy $S_2 = -\log \Tr \rho^2$, although we will also prove various statements on the von Neumann entropy $S_1 = -\Tr \rho\log\rho$ as well. The R\'enyi-2 entropy takes the elegant form $S_2 = \frac{1}{2}\log\det \sigma = \frac{1}{2}\Tr \log \sigma$, where $\sigma$ is the covariance matrix of $\rho$ \cite{serafini_quantum_2017}. For the Gaussian state generated by the unitary $U \in \U(n)$ acting on the squeezed product state with squeezing strength $s$ on each mode, the R\'enyi-2 entropy of the subsystem of $k$ modes is denoted by $S_2(U)$, and the von Neumann entropy by $S_1(U)$, where the dependence of $S_1(U)$ and $S_2(U)$ on $r$, $n$, and $s$ is implicit. We will derive statistical properties of $S_2(U)$ and $S_1(U)$ in the asymptotic limit $n\to\infty$.

Our main results involve the R\'enyi-2 entropy, and the following proposition allows us to use many of the results on the R\'enyi-2 entropy to bound the von Neumann entropy. Furthermore, we will also make use of the maximum R\'enyi-2 entropy to prove our later results on the Page correction.

\medskip

\begin{proposition}
\label{prop:s1-bound}
Let $s \in \bbR$ and $r\in[0,1]$. Then for each $j\in\set{1,2}$,
\begin{equation}
    \label{eq:entropy-maximum}
    \max_{U\in\U(n)}S_j(U) = n\min(r,1-r) h_j(\cosh(2s)),
\end{equation}
where $h_1(x) = \frac{x+1}{2}\log\frac{x+1}{2}-\frac{x-1}{2}\log\frac{x-1}{2}$ and $h_2(x) = \log x$. Furthermore, for any fixed $U \in \U(n)$, $S_1(U) \geq S_2(U)$ and
\begin{equation}
    \label{eq:vonNeumann-upper-bound}
    S_1(U) < S_2(U) + n\min(r,1-r)(1-\log 2).
\end{equation}
\end{proposition}

\medskip

The full proof of \cref{prop:s1-bound} is given in \cref{ap:entropy-bounds}. A tighter version of \cref{eq:vonNeumann-upper-bound} was originally derived in Ref.~\cite[Eq.~15]{adesso_strong_2016}, but we will only need this weaker version. For completeness, we provide a different proof of \cref{eq:vonNeumann-upper-bound}. \cref{eq:entropy-maximum} is perhaps implicit in various results in Refs.~\cite{serafini_canonical_2007,serafini_teleportation_2007,khannaMaximumEntanglementSqueezed2007,serafini_quantum_2017,stanisicGeneratingEntanglementLinear2017}, but we have not found it directly stated anywhere. 
The lower bound $S_1(U) \geq S_2(U)$ is a general property of the R\'enyi entropies \cite{nielsen_quantum_2010,wilde_quantum_2017}, whereas the upper bound holds only for Gaussian states. While trivial upper bounds also exist in the general case, \cref{eq:vonNeumann-upper-bound} is tighter. Intuitively one may view this as an extension of the fact that for Gaussian states, the R\'enyi-2 and von Neumann entropies share many useful properties, such as strong subadditivity and others mentioned in \cref{sec:intro}.

\section{Expectation value}
\label{sec:expval}

Our first results concern the R\'enyi-2 Page curve, which is the expectation value of the R\'enyi-2 entanglement entropy as a function of the partition size ratio $r = k/n$ and the squeezing strength $s$. Recall that the dependence of $S_2(U)$ on $r$, $n$, and $s$ is implicit. We find an exact formula as an infinite series for the Page curve in the limit $n\to\infty$.

\medskip

\begin{theorem}[R\'enyi-2 Page curve]
\label{thm:page-curve}
Fix $s \in \bbR$ and $r \in [0,1]$. Let $C_\ell\coloneqq \frac{1}{\ell+1}\binom{2\ell}{\ell}$ be the $\ell^{\rm th}$ Catalan number, and let ${}_{2}F_1$ be the hypergeometric function \cite{petkovsek1996,bailey_generalized_1964,zudilin_hypergeometric_2019,slater_generalized_1966,wolfram_hypergeometric2f1}. Then
\begin{widetext}
\begin{align}
    \lim_{n\to\infty} \Expval_{U \in \U(n)} \frac{1}{n} S_2(U)
    &= r \log \cosh(2s) - \sum _{\ell=1}^\infty r^{\ell+1} \tanh^{2 \ell}(2 s) \frac{C_\ell}{2 \ell} \, _2F_1(1-\ell,\ell;\ell+2;r)\label{eq:S2series1}\\
    &= r \log \cosh(2s) + \sum_{d=2}^\infty r^d \sum_{\ell=\ceil{d/2}}^{d-1} (-1)^{d-\ell} \tanh^{2\ell}(2s) \frac{(2 \ell-1)!}{d (d-1) (d-\ell-1)! (2\ell-d)! \ell!}. \label{eq:S2series2}
\end{align}
\end{widetext}
This function is symmetric under $r \mapsto 1-r$, and hence the formula holds when $r$ is replaced with $\min(r,1-r)$. Furthermore, asymptotically in $n$, $\Expval_{U\in\U(n)} S_2(U) = n \alpha(s,r)-\lambda(s,r) + \littleo{1}$, where $\alpha(s,r)$ is precisely $\lim_{n\to\infty}\Expval_{U\in\U(n)} \frac{1}{n}S_2(U)$ given above, $\lambda(s,r)$ is independent of $n$, and $\littleo{1}$ denotes terms that go to zero as $n\to\infty$.
\end{theorem}

\medskip

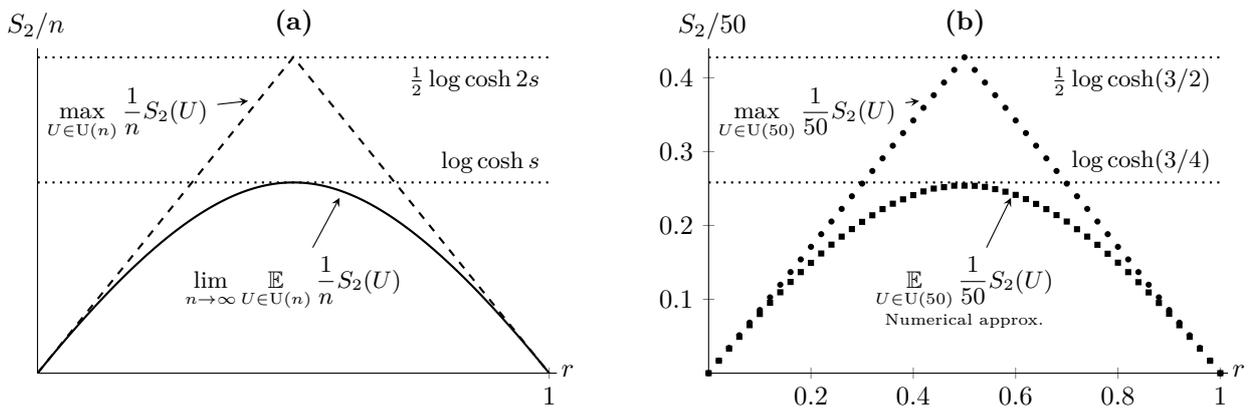
\begin{figure*}
    \centering
    \begin{minipage}{.495\textwidth}
    \begin{tikzpicture}
    \begin{axis}[
            width=\textwidth,
            height=.7\textwidth,
            xlabel={$r$},
            ylabel={$S_2/n$},
            title={\textbf{(a)}},
            samples=3,
            domain=0:1,
            xtick={1},
            ytick={0},
            ymin=0, ymax=.44,
            xmin=0, xmax=1.015,
            restrict y to domain =0:1,
            axis y line=middle,
            axis x line=middle,
            x axis line style=-,
            y axis line style=-,
            x label style={xshift=10,yshift=-5},
            y label style={xshift=-15,yshift=17},
            title style={xshift=-1.2,yshift=-5}
        ]
    
        \addplot[black, thick, dotted] plot (\x, {ln(cosh(3/4))+(1/2)*ln(1+pow(tanh(3/4),2))});
        \node at (axis cs:1,.4275) [anchor=north east] {\small $\frac{1}{2}\log \cosh 2s$};
        
        \addplot[black, thick, dotted] plot (\x, {ln(cosh(3/4))});
        \node at (axis cs:1,.258) [anchor=south east] {\small $\log \cosh s$};
        
        \addplot[black, thick, dashed] plot (\x, {min(\x,1-\x)*ln(cosh(3/2))});
        \node (max) at (axis cs: 0,.35) [anchor=west] {\small $\displaystyle\max_{U \in \U(n)} \frac{1}{n}S_2(U)$};
        \node (maxcurve) at (axis cs: .43,.37) {};
        \draw[->,>=stealth](max)--(maxcurve);

        \addplot[black, thick] table [x=r, y=Squal]{\relativepath qualitative.dat};
        \node (page) at (axis cs:.5,.17) [anchor=north] {\small $\displaystyle\lim_{n\to\infty} \displaystyle\Expval_{U\in\U(n)}\frac{1}{n}S_2(U)$};
        \node (curve) at (axis cs:.6,.25){};
        \draw[->,>=stealth](page)--(curve);
    
    \end{axis}
    \end{tikzpicture}
    \end{minipage}
    \hfill
    \begin{minipage}{.495\textwidth}
    \begin{tikzpicture}
    \begin{axis}[
            width=\textwidth,
            height=.7\textwidth,
            xlabel={$r$},
            ylabel={$S_2/50$},
            title={\textbf{(b)}},
            samples=3,
            domain=0:1,
            xtick={.2,.4,.6,.8,1},
            ytick={0,.1,.2,.3,.4},
            ymin=0, ymax=.44,
            xmin=0, xmax=1.015,
            restrict y to domain =0:1,
            axis y line=middle,
            axis x line=middle,
            x axis line style=-,
            y axis line style=-,
            x label style={xshift=10,yshift=-5},
            y label style={xshift=-15,yshift=17},
            title style={xshift=-1.2,yshift=-5}
        ]
    
        \addplot[black, thick, dotted] plot (\x, {ln(cosh(3/4))+(1/2)*ln(1+pow(tanh(3/4),2))});
        \node at (axis cs:1,.4275) [anchor=north east] {\small $\frac{1}{2}\log \cosh(3/2)$};
        
        \addplot[black, thick, dotted] plot (\x, {ln(cosh(3/4))});
        \node at (axis cs:1,.258) [anchor=south east] {\small $\log \cosh (3/4)$};

        \addplot[black, thick, only marks, mark size=.75pt] table [x=r, y=Smax]{\relativepath quantitative.dat};
        \node (max) at (axis cs: 0,.35) [anchor=west] {\small $\displaystyle\max_{U \in \U(50)} \frac{1}{50}S_2(U)$};
        \node (maxcurve) at (axis cs: .43,.37) {};
        \draw[->,>=stealth](max)--(maxcurve);

        \addplot[black, thick, only marks, mark size=.75pt, mark=square*] table [x=r, y=Squan]{\relativepath quantitative.dat};
        \node (page) at (axis cs:.5,.17) [anchor=north] {\small $\substack{\displaystyle\Expval_{U\in\U(50)}\frac{1}{50}S_2(U)\\\text{Numerical approx.}}$};
        \node (curve) at (axis cs:.6,.25){};
        \draw[->,>=stealth](page)--(curve);
    
    \end{axis}
    \end{tikzpicture}
    \end{minipage}

    \caption{(a) Exact results for the R\'enyi-2 Page curve from \cref{prop:s1-bound,thm:page-curve,cor:special-page-curve}. (b) Numerical simulations of the R\'enyi-2 Page curve for $n=50$ modes and squeezing $s = 3/4$. We plot the values for $r=k/50$ for each $k\in\set{0,1,\dots,50}$. We perform the simulation by generating random unitary matrices and doing the matrix multiplication described in \cref{ap:gaussian-states}. We provide the code for this simulation on GitHub \cite{joseph_t_iosue_glo_2022}.}
    \label{fig:page-curve}
\end{figure*}

We plot the analytic Page curve for $s=3/4$ in \cref{fig:page-curve}(a), and we confirm our results numerically in \cref{fig:page-curve}(b). The proof of \cref{thm:page-curve}, given in \cref{ap:page}, primarily uses two ingredients. The first ingredient is the asymptotic form of the Weingarten calculus for integrating over the unitary group with respect to the Haar measure \cite{weingarten_asymptotic_1978,collins_moments_2002}. From this we get an equation for the Page curve that is initially daunting. The second main ingredient is the fact that the Page curve must be symmetric under $r \mapsto 1-r$ since the global state on the $n$ modes is pure \cite{nielsen_quantum_2010}. Quite miraculously, this fact is enough to simplify the equation for the Page curve and arrive at \cref{thm:page-curve}.

The Page curve derived in \cref{thm:page-curve} can be written as 
\begin{equation}
    \lim_{n\to\infty} \Expval_{U\in\U(n)} \frac{1}{n}S_2(U) = \sum_{\ell=1}^\infty \frac{\tanh^{2\ell}(2s)}{2\ell} G_\ell(r),
\end{equation}
where $G_\ell(r)\coloneqq r-f_\ell(r)$ and $f_\ell(r)$ is a polynomial of degrees $\ell+1$ through $2\ell$ in $r$ (given in \cref{eq:analytic-W} in \cref{ap:page}). Polynomials $G_\ell(r)$ of this form are uniquely determined by the requirement that $G_\ell(r) = G_\ell(1-r)$, which ensures that the R\'enyi-2 entropy of a subsystem is equal to that of its complement since we are considering pure states. It is from this requirement that we ultimately derive the Page curve. We show that the resulting $G_\ell(r)$ can be understood as a good approximation to $m(r)\coloneqq \min(r,1-r)$ from below, which we will call the $\ell^{\rm th}$ approximation.
Indeed, the approximation is especially good near the endpoints $r=0$ and $r=1$, where the first $\ell$ derivatives of $G_\ell(r)$ match those of $m(r)$. As $\ell\to\infty$, the approximation becomes better and better such that $\lim_{\ell\to\infty}G_\ell(r) = m(r)$. 

This provides an interpretation of the derived form of the Page curve. The strength of the squeezing $s$ determines the weight that the Page curve has on the $\ell^{\rm th}$ approximation to $m(r)$. For small squeezing, only low order approximations contribute, with the most dominant contribution being the parabolic shape $G_1(r) = r(1-r)$. When the squeezing is increased, there is more contribution from higher order approximations, giving the Page curve more of the triangle shape of $m(r)$. We see a manifestation of this interpretation as
\begin{align}
    \lim_{s \to 0} &\lim_{n \to \infty} \Expval_{U \in \U(n)} \frac{1}{s^2 n} S_2(U) = 2r(1-r), \\
    \lim_{s\to\infty} &\lim_{n\to\infty} \Expval_{U \in \U(n)} \frac{1}{s n} S_2(U) = 2 \min(r,1-r).
\end{align}
Meanwhile, from \cref{eq:entropy-maximum}, the maximal R\'enyi-2 entropy is $\max_U \frac{1}{n}S_2(U) = m(r)\log\cosh(2s)$. As stated, near the endpoints $r=0$ and $r=1$, $G_\ell(r)$ is a very good approximation to $m(r)$. Thus, regardless of the squeezing strength, when the subsystem size $k = rn$ is small (or when its complement is small), the average entanglement is very close to maximal.

Unfortunately, we are unable to simplify the infinite sum for general $r$ in \cref{thm:page-curve} further. However, the Page curve can be fully simplified at $r=1/2$ --- which is where the maximum for a fixed $s$ occurs --- to $\log\cosh s$. Indeed, we find that 
\begin{equation}
    G_\ell(1/2) = \frac{1}{2}\parentheses{1-4^{-\ell}\binom{2\ell}{\ell}}.
\end{equation}
From this and the maximum R\'enyi-2 entropy from \cref{eq:entropy-maximum}, we also find the exact expression for the Page correction at $r = 1/2$ to be $\frac{1}{2}\log\pargs{1+\tanh^2 s}$. Let us formally state these results as follows.

\medskip

\begin{corollary}
\label{cor:special-page-curve}
For a fixed $s \in \bbR$ when $r=1/2$,
\begin{equation}
    \lim_{n\to\infty} \Expval_{U \in \U(n)} \frac{1}{n}S_2(U) = \log\cosh s,
\end{equation}
and the Page correction is
\begin{equation}
\begin{aligned}
    \lim_{n\to\infty} \frac{1}{n}&\parentheses{\max_{U \in \U(n)} S_2(U) - \Expval_{U \in \U(n)} S_2(U)}\\
    &= \frac{1}{2}\log\pargs{1+\tanh^2 s}.
\end{aligned}
\end{equation}
\end{corollary}

\medskip

The proof of \cref{cor:special-page-curve}, given in \cref{ap:special-page}, is a straightforward consequence of a simplification of the hypergeometric function $\, _2F_1(a,1-a;c;1/2)$ in terms of gamma functions due to Bailey's theorem \cite{petkovsek1996,bailey_generalized_1964,zudilin_hypergeometric_2019,slater_generalized_1966,wolfram_hypergeometric2f1}. Altogether, we have derived the exact formula [\cref{eq:S2series1,eq:S2series2}] for the R\'enyi-2 page curve in the regime of equal squeezers as a series in $\tanh^2(2s)$ and $r$. In the special case when $r=1/2$, we simplified the series to obtain an exact value of $\log \cosh s$. From this, we derived the Page correction, or information of the subsystem, at $r=1/2$ to be exactly $\frac{1}{2}\log\pargs{1+\tanh^2 s}$. Furthermore, since the Page curve is concave in $r$ while the maximum entropy is linear, this correction is maximized at $r=1/2$. A summary of the Page curve results thus far is provided in \cref{fig:page-curve}.

\medskip

We now shift our attention to the constant term in the Page curve. \cref{thm:page-curve} states that asymptotically in $n$, $\Expval_{U\in\U(n)} S_2(U) = n\alpha(s,r) - \lambda(s,r) + \littleo{1}$, and it provides the exact expression for $\alpha(s,r)$. We further find the following result for $\lambda$, which is proven in \cref{ap:const}.

\medskip

\begin{proposition}
\label{prop:const}
Fix $s \in \bbR$ and $r \in [0,1]$. Then,
\begin{equation}
    \lambda(s,r) = -\frac{1}{8} \log \pargs{ 1-4 r (1-r) \tanh ^2(2 s) }.
\end{equation}
\end{proposition}

\medskip

Note that $\lambda(s,r) \geq 0$ for all $s$ and $r$, and therefore $\Expval_U S_2(U) \leq n \alpha(s,r)$ asymptotically. At $r=1/2$, this simplifies to $\lambda(s,1/2) = \frac{1}{4}\log\cosh(2s)$.

The proof of \cref{prop:const} is very similar to the proof of \cref{thm:page-curve}, again crucially using the symmetry $r\mapsto 1-r$ to simplify the expression coming from the Weingarten function. However, there is one substantial difficulty in the proof of \cref{prop:const} that does not occur in that of \cref{thm:page-curve}. In this case, the symmetry $r \mapsto 1-r$ is \textit{almost} enough to fully determine $\lambda$; however, there are constants that cannot be determined by the symmetry alone. We derive expressions for these constants that are complicated sums involving permutations and Catalan numbers. Using objects that arise in bioinformatics when studying gene orders, called \textit{breakpoint graphs} \cite{Alexeev2017hultman}, Ref.~\cite{alekseyev} evaluates these sums, hence completing the proof of \cref{prop:const}. These constants are discussed more in \cref{ap:const}.

In summary, we have derived an explicit form of the asymptotic R\'enyi-2 entropy Page curve for all $r$ and $s$ up to corrections that vanish as $n\to\infty$. 

\section{Variance and typicality}
\label{sec:variance}

Next, we shift our attention to the variance of the R\'enyi-2 entropy so as to make statements about typicality of entanglement; that is, how different a random state's entanglement is from the average. Using the results for the R\'enyi-2 entropy, we will also be able to prove some weaker results for the von Neumann entropy. The typicality results presented below are summarized in \cref{tab:typicality}.

Typicality is of interest because it characterizes the applicability of statistical averages. Indeed, statistical mechanics often relies on quantities being typical so that thermodynamic average quantities, such as average energy and average pressure, can accurately represent their true values.
In order to quantify the deviation from average, we consider two measures of deviation corresponding to multiplicative and additive distance. If the multiplicative distance between a quantity and its average vanishes in the thermodynamic limit, then that quantity is called \textit{weakly typical}. If the additive distance vanishes in this limit, then that quantity is called \textit{strongly typical}. With this intuition, we now formally define weak and strong typicality following Ref.~\cite{dahlsten_entanglement_2014}.

\medskip

\begin{definition}[Typicality]
\label{def:typicality}
Let $S$ be a nonnegative random variable on the unitary group $\U(n)$, and denote its value at $U\in\U(n)$ by $S(U)$. $S$ is called \emph{weakly typical} if for any constant $\epsilon > 0$,
\begin{equation}
    \lim_{n\to\infty} \Pr{U\in\U(n)}{\abs{\frac{S(U)}{\Expval_{V\in\U(n)} S(V)} - 1} < \epsilon} = 1.
\end{equation}
$S$ is called \emph{strongly typical} if for any constant $\epsilon > 0$,
\begin{equation}
    \lim_{n\to\infty} \Pr{U\in\U(n)}{\abs{S(U) - \Expval_{V\in\U(n)} S(V)} < \epsilon} = 1.
\end{equation}
\end{definition}

\medskip

\setlength{\tabcolsep}{12pt}

\begin{table*}[ht]
    \centering
    \begin{tabular}{ccccc}
         \toprule
         & & $k \in \bigTheta{n}$ & $k \in \littleo{n}$ & $k \in \littleo{n^{1/3}}$ \cite{fukuda_typical_2019} \\
        \midrule
         \textbf{Equal} & \textit{R\'enyi-2} & weak & strong & strong\\
         \textbf{squeezing} & \textit{von Neumann} & ? & weak & strong \\\\
         \textbf{Unequal} & \textit{R\'enyi-2} & ? & weak${}^\ast$ & strong \\
         \textbf{squeezing} & \textit{von Neumann} & ? & weak${}^\ast$ & strong \\
        \bottomrule
    \end{tabular}
    \caption{
    A summary of the current status of rigorous results on typical entanglement in Gaussian bosonic systems, where in this figure we assume single-mode squeezing parameters that are independent of $n$. Strong and weak typicality are defined in \cref{def:typicality}. Note that ``weak${}^\ast$'' indicates that the result is not fully proven, but depends on \cref{conj:entanglement-derivative}. Where we say ``weak'', we have not ruled out the possibility that the typicality is also strong. 
    The total number of modes is denoted by $n$, and $0\leq k \leq n$ is the number of modes in the subsystem. ``Equal squeezing'' refers to the case when each mode is initially squeezed with the same strength, whereas ``unequal squeezing'' refers to the general case when each mode can be squeezed independently. The two leftmost columns come from \cref{cor:renyi2-equal-squeeze-typical,cor:vonNeumann-equal-squeezing-typical,rmk:unequal-squeeze-typical}. The rightmost column all follows from the results of Ref.~\cite{fukuda_typical_2019}. Prior to Ref.~\cite{fukuda_typical_2019}, Refs.~\cite{serafini_canonical_2007,serafini_teleportation_2007} proved strong typicality in the regime $k \in \bigO{1}$.
    }
    \label{tab:typicality}
\end{table*}

If $\Expval_{U\in \U(n)} S(U)$ does not decay as $n\to\infty$, then strong typicality clearly implies weak typicality. 
We will be concerned with typicality of $S_1$ and $S_2$, the von Neumann and R\'enyi-2 entropy respectively. One can compute the variance of the entropy over the Haar measure and then apply Chebyshev's inequality to obtain typicality results. Therefore, we now focus on the variance. We first find the general form of the asymptotic R\'enyi-2 variance in the equal squeezing regime and show that it is independent of $n$. Recall that the dependence of $S_2(U)$ on $r$, $n$, and $s$ is implicit.

\medskip

\begin{theorem}[R\'enyi-2 variance]
\label{thm:renyi2-variance}
Let $s \in \bbR$ and $r \in [0,1]$. Then
\begin{equation}
    \lim_{n\to\infty}\variance_{U \in \U(n)} S_2(U) = \sum_{d=2}^\infty \omega^{(d)} \tanh^{2d}(2s)  \parentheses{r(1-r)}^d,
\end{equation}
where $\omega^{(d)} \in \bbQ$ is some number that depends only on $d$. In particular, $\omega^{(2)}=1/2$.
\end{theorem}

\medskip

The proof of this theorem, given in \cref{ap:variance}, is very similar to the proof of \cref{thm:page-curve}. To prove this theorem, we again crucially use that $\variance_U S_2(U)$ must be symmetric under $r \mapsto 1-r$ since the full state on the $n$ modes is pure. Interestingly, in contrast to \cref{thm:page-curve} where we found that the Page curve grows linearly with $n$, the asymptotic variance is independent of $n$. Indeed, \cref{thm:renyi2-variance} is the bosonic analogue of the result that the variance for fermionic Gaussian states is asymptotically constant \cite{bianchi_page_2021}.

From \cref{thm:renyi2-variance}, we find typicality in certain regimes as an immediate corollary. In particular, since the variance is independent of $n$ while the average grows with $n$, the entanglement is always weakly typical. Furthermore, the variance is asymptotically zero if $s \in \littleo{1}$ and/or $r \in \littleo{1}$, and typicality is therefore strong in those regimes. Altogether, \cref{thm:renyi2-variance} and Chebyshev's inequality lead to the following corollary.

\medskip

\begin{corollary}
\label{cor:renyi2-equal-squeeze-typical}
The R\'enyi-2 entropy is weakly typical for any $s \in \bbR$ and $r\in [0,1]$. Furthermore, if the subsystem size $k=rn$ scales as $k \in \littleo{n}$, or if the squeezing scales as $s \in \littleo{1}$, then the R\'enyi-2 entropy is strongly typical, and $\Pr{U \in \U(n)}{\abs{S_2(U) - \Expval_V S_2(V)} < \epsilon} \geq 1 - \frac{\bigOmega{1/n}}{\epsilon^2}$. On the other hand, if $r$ and $s$ are constant in $n$, then
$\lim_{n\to\infty}\Pr{U \in \U(n)}{\abs{S_2(U) - \Expval_V S_2(V)} < \epsilon} \geq 1 - \frac{c(r,s)}{\epsilon^2}$,
where $c(r,s)$ is a constant independent of $n$ (but depends on $r$ and $s$).
\end{corollary}

\medskip

In summary, we have shown weak typicality in the R\'enyi-2 entropy for all $s$ and $k$ and strong typicality whenever $k \in \littleo{n}$. On the other hand, if $s$ and $r$ do not tend to zero with increasing $n$, the variance converges to a constant value independent of $n$. Hence, we cannot use Chebyshev's inequality to prove strong typicality in this case. Notably, an asymptotically constant variance does not necessarily imply an absence of strong typicality either --- that is, the probability that the entropy deviates from its average can scale as $1/n^2$, but the entropy can deviate by an amount proportional to $n$ thus resulting in a constant variance. We are therefore unable to make a definitive statement about strong typicality in the $r \in \bigTheta{1}$ and $s \in \bigOmega{1}$ case.

In the next corollary, we will use these results to address typicality as measured by the von Neumann entropy. Specifically, we will use \cref{prop:s1-bound} to show weak typicality of the von Neumann entropy as long as the subsystem size scales sublinearly with the system size.

\medskip

\begin{corollary}
\label{cor:vonNeumann-equal-squeezing-typical}
Let $s \in \bbR$ and $r \in [0,1]$. If the subsystem size $k=rn$ scales as $k \in \littleo{n}$, then the von Neumann entropy is weakly typical.
\end{corollary}
\begin{proof}
Using \cref{prop:s1-bound} to upper bound $\Expval_U \frac{1}{n^2}S_1(U)^2$ and to lower bound $\parentheses{\Expval_U \frac{1}{n}S_1(U)}^2$, we find that
\begin{equation}
    \label{eq:s1-var-bound}
    \begin{aligned}
        \lim_{n\to\infty}&\frac{1}{n^2} \variance_{U\in\U(n)} S_1(U) < \lim_{n\to\infty}\frac{1}{n^2} \variance_{U\in\U(n)} S_2(U) \\
        &+ 2r\log(\e/2) \lim_{n\to\infty} \frac{1}{n}\Expval_{U\in\U(n)} S_2(U)\\
        &+ r^2 \log^2(\e/2).
    \end{aligned}
\end{equation}
From \cref{thm:page-curve}, the first term in the right-hand side is always zero. Furthermore, from \cref{thm:renyi2-variance}, the second term is $\bigO{r^2}$. Hence, $\variance_U S_1(U)/n^2 \in \bigO{r^2}$, which is zero as $n\to\infty$ when $r \in \littleo{1}$.
\end{proof}

\medskip

We again emphasize that our typicality results thus far, which are summarized in \cref{tab:typicality}, only apply to the case when the initial squeezing strength on each mode is the same. On the other hand, Ref.~\cite{fukuda_typical_2019} proves strong typicality when $k\in \littleo{n^{1/3}}$ for both the von Neumann and R\'enyi-2 entropies in the general case when squeezing strengths can be different on different modes. It is for this reason that our results do not supplant those of Ref.~\cite{fukuda_typical_2019} in general, but rather only in certain regimes. To the best of our knowledge, our results are the first to address the $k\in\bigOmega{n^{1/3}}$ regime.

Ultimately, we would like to determine in exactly which regimes strong and weak typicality occur and do not occur. Our results thus far almost complete the story for the regime of equal squeezing when typicality is measured with the R\'enyi-2 entropy, since we have proven strong typicality when $k\in \littleo{n}$ and weak when $k \in \bigTheta{n}$; the missing piece is whether or not typicality is strong when $k \in \bigTheta{n}$. However, the story is even more incomplete for the regime of equal squeezing when typicality is measured with the von Neumann entropy, though we made some progress by proving that typicality is at least weak whenever $k \in \littleo{n}$. In this regime, the best known result for strong typicality is when $k\in \littleo{n^{1/3}}$ as proven in Ref.~\cite{fukuda_typical_2019}. Indeed, this is also currently the best known result in the regime of unequal squeezing. In \cref{sec:unequal-squeezing}, we will use our results on equal squeezing typicality to add to the story for unequal squeezing.

\section{Generalizing to unequal squeezing}
\label{sec:unequal-squeezing}

So far, we have considered the restricted setting where each mode $i$ is initially squeezed with strength $s_i = s$ for some $s \in \bbR$. We now generalize by allowing the squeezing strengths to be different on each mode. As such, for the remainder of this section, the squeezings will be $(s_1,\dots, s_n)$, where each $s_i \in \bbR$, and $s_{\rm max}$ and $s_{\rm min}$ are defined as $s_{\rm max} \coloneqq \max_i \abs{s_i}$ and $s_{\rm min} \coloneqq \min_i \abs{s_i}$.

To begin, we focus on the R\'enyi-2 Page curve in the regime of unequal squeezing. When squeezing is small, we can utilize the equal squeezing Page curve in \cref{thm:page-curve} to compute the Page curve for unequal squeezing. Recall that the dependence of $S_2(U)$ on $r$, $n$, and $s$ is implicit.

\medskip

\begin{corollary}
Let each $s_i \in \bbR$ and $r \in [0,1]$. Then, as $n\to\infty$,
\begin{equation}
    \Expval_{U\in\U(n)}S_2(U) = 2r(1-r)\sum_{i=1}^n s_i^2 + \bigO{rn s_{\rm max}^4}.
\end{equation}
\end{corollary}
\begin{proof}
The R\'enyi-2 entropy of a Gaussian state with covariance matrix $\sigma$ is proportional to $\log \det \sigma$. $\log$ is real analytic and $\det \sigma$ is analytic in the parameters of the initial covariance matrix $\sigma_0$. Hence, $S_2$ can be written as a power series in $s_i$. There exists a passive Gaussian unitary, specifically a product of one-mode phase shifters, that acts on $\sigma_0$ via the transformation $s_i \to -s_i$. Hence, by the translational invariance of the Haar measure --- that is, the invariance of the Haar measure under the application of any fixed unitary --- the power series must be an expansion in $s_i^2$. Furthermore, there exists a passive Gaussian unitary, specifically a product of beamsplitters, that acts on $\sigma_0$ via the transformation $s_i \to s_{\tau(i)}$ for some permutation $\tau$ of $n$ elements. Therefore, again by the translational invariance of the Haar measure, the power series must be symmetric under $s_i \to s_{\tau(i)}$. It follows that the $\bigO{s_i^2}$ term must be of the form $g(r,n)\sum_i s_i^2$ for some function $g(r, n)$. When all $s_i$ are equal, the power series must reduce to \cref{thm:page-curve}, which fixes $g(r, n)$ to be $2r(1-r)$.

The next term is $\bigO{s_{\rm max}^4}$, but what is the $k$ dependence? We will show that the $k$ dependence is at most linear, proving that the remaining terms in the power series are $\bigO{k s_{\rm max}^4}$.
Recall that $S_2(U) = \frac{1}{2}\Tr \log \sigma(U)$ where $\sigma(U)$ is the covariance matrix for the state generated by $U$ from the initial product squeezed state $\sigma_0$. Since the $\log$ function is concave, Jensen's inequality implies that $\Expval_U S_2(U) \leq \frac{1}{2}\Tr \log \Expval_U \sigma(U)$. $\Expval_U \sigma(U)$ is calculated in \cref{eq:sigma-U-trace} in \cref{ap:gaussian-states} to be $\nu \bbI_{2k\times 2k}$, where $\nu = \frac{1}{n}\sum_i \cosh(2s_i) \leq \cosh(2s_{\rm max})$. Therefore, $\Expval_U S_2(U) \in \bigO{k \log \nu}$.
\end{proof}

\medskip

One particularly interesting application of this corollary is when each $s_i \in \bigO{1/\sqrt{n}}$. In this case, the average total number of bosons in the $n$ modes is $N=\sum_i \sinh^2(s_i) \in \bigO{1}$. Thus, when one considers a constant number $N$ of bosons in the system as the number of modes is taken to infinity, one finds the Page curve to be $2r(1-r)N$.

\medskip

We now shift our focus to entanglement typicality in the regime of unequal squeezing. The results described in the remainder of this section are summarized in \cref{tab:typicality}. 
Ideally, we would like to make further statements about entanglement in the regime of unequal squeezing by using our previous results on equal squeezing. One way to potentially proceed is to use the equal squeezing results to bound the unequal squeezing quantities. Intuitively, we expect that $\Expval_U S_2(U)$ for arbitrary unequal squeezing is upper bounded by $\Expval_U S_2(U)\rvert_{\text{all }s_i=s_{\rm max}}$ and lower bounded by $\Expval_U S_2(U)\rvert_{\text{all }s_i=s_{\rm min}}$. In other words, by \textit{increasing} all of the squeezing strengths until they all are equal, the average entanglement will \textit{increase}. In this spirit, we make the following conjecture.

\medskip

\begin{conjecture}
\label{conj:entanglement-derivative}
Let $r\in[0,1]$, $s_i \in \bbR$, and $n \in \bbN$. Then, for any $i$,
\begin{align}
    &0 \leq \frac{\partial}{\partial (s_i^2)} \parentheses{\Expval_{U\in\U(n)} \frac{1}{n}S_2(U)}, \text{  and}\\
    &0 \leq \frac{\partial}{\partial (s_i^2)} \parentheses{\Expval_{U\in\U(n)} \frac{1}{n^2}S_2(U)^2}.
\end{align}
\end{conjecture}

\medskip

\cref{conj:entanglement-derivative} seems intuitive --- by increasing the magnitude of any individual squeezing strength, the number of bosons in the system increases, and therefore it would seem surprising for the average entanglement to decrease. We note that, somewhat counterintuitively, one can find explicit unitaries and squeezing configurations for which $\frac{\partial}{\partial (s_i^2)} S_2(U) < 0$ (we provide an example in our code repository \cite{joseph_t_iosue_glo_2022}), and hence the presence of the $\Expval_U$ is necessary for the conjecture. Despite its intuitiveness, we have been unable to rigorously prove \cref{conj:entanglement-derivative}. The derivative of $S_2(U) = \frac{1}{2} \log \det \sigma(U)$ is $\frac{1}{2}\Tr\pargs{\sigma(U)^{-1} \frac{\partial \sigma(U)}{\partial(s_i)^2}}$. The difficulty in computing the expectation value over $U\in\U(n)$ arises due to the presence of the inverse $\sigma(U)^{-1}$.

Nonetheless, under the assumption that \cref{conj:entanglement-derivative} is true, we can immediately upper and lower bound the R\'enyi-2 Page curve for an arbitrary squeezing configuration $(s_1,\dots,s_n)$ by using \cref{thm:page-curve} with $s=s_{\rm max}$ and $s=s_{\rm min}$ respectively,
\begin{align}
    &\lim_{n\to\infty}\Expval_{U\in\U(n)}\frac{1}{n} S_2(U) \geq \lim_{n\to\infty}\Expval_{U\in\U(n)}\frac{1}{n} S_2(U)\big\rvert_{s_i=s_{\rm min}}, \label{eq:conjecture-lowerbound}\\
    &\lim_{n\to\infty}\Expval_{U\in\U(n)}\frac{1}{n} S_2(U) \leq \lim_{n\to\infty}\Expval_{U\in\U(n)}\frac{1}{n} S_2(U)\big\rvert_{s_i=s_{\rm max}}.
\end{align}
Furthermore, the conjecture implies that
\begin{equation}
    \label{eq:conjecture-upperbound}
    \begin{aligned}
        \lim_{n\to\infty}&\Expval_{U\in\U(n)}\frac{1}{n^2} S_2(U)^2 \\
        &\leq \lim_{n\to\infty}\Expval_{U\in\U(n)}\frac{1}{n^2} S_2(U)^2\big\rvert_{s_i=s_{\rm max}}.
    \end{aligned}
\end{equation}
From this, we can also make statements on weak typicality for unequal squeezers by bounding the variance. The variance is $\Expval_U S_2(U)^2 - (\Expval_U S_2(U))^2$. We can therefore upper bound the variance by upper bounding the first term with \cref{eq:conjecture-upperbound} and lower bounding the second term with \cref{eq:conjecture-lowerbound}.

\medskip

\begin{remark}
\label{rmk:unequal-squeeze-typical}
Let $r\in [0,1]$ and $s_i \in \bbR$. If \cref{conj:entanglement-derivative} is true and if the subsystem size $k=rn \in \littleo{n}$ or $s_{\rm max} \in \littleo{1}$, then the R\'enyi-2 entropy is weakly typical. Similarly, if $k \in \littleo{n}$, then the von Neumann entropy is weakly typical. 
\end{remark}
\begin{proof}
We can bound the variance of the R\'enyi-2 entropy as
\begin{equation}
    \begin{aligned}
        \lim_{n\to\infty}&\frac{1}{n^2}\variance_U S_2(U) \leq \lim_{n\to\infty} \frac{1}{n^2}\\
         &\times \bigg[\Expval_U \parentheses{\frac{1}{n}S_2(U)\rvert_{\text{all } s_i = s_{\rm max}}}^2\\
         &\phantom{\times\bigg[} - \parentheses{\Expval_U \frac{1}{n}S_2(U)\rvert_{\text{all } s_i = s_{\rm min}}}^2 \bigg].
    \end{aligned}
\end{equation}
From \cref{thm:renyi2-variance}, the $\Expval_U$ can be brought inside the parentheses in the first term. Then, the right hand side can be computed using \cref{thm:page-curve}, which gives $\lim_{n\to\infty}\frac{1}{n^2}\variance_U S_2(U) \in \bigO{r^2 s_{\rm max}^4}$. Hence $S_2$ is weakly typical if $r \in \littleo{1}$ or $s_{\rm max} \in \littleo{1}$.

For the von Neumann entropy, we again make use of \cref{eq:s1-var-bound}, which gives $\lim_{n\to\infty}\frac{1}{n^2}\variance_U S_1(U) \in \bigO{r^2 s_{\rm max}^2}$, which goes to zero if $r \in \littleo{1}$.
\end{proof}

\medskip

In summary, we have used our results from \cref{sec:expval,sec:variance} on entanglement in the equal squeezing regime to prove various statements in the unequal squeezing regime. In particular, under the assumption that \cref{conj:entanglement-derivative} is true, we prove both the R\'enyi-2 and von Neumann entropies are weakly typical whenever the subsystem size $k$ scales as $k \in \littleo{n}$. The best known result for the presence of strong typicality in the unequal squeezing regime is when $k\in \littleo{n^{1/3}}$ as proven in Ref.~\cite{fukuda_typical_2019}. The current status of rigorous results on typical entanglement in Gaussian bosonic systems is summarized in \cref{tab:typicality}.

\section{Conclusion}
\label{sec:conclusion}

In this work, we studied the average and variance of the R\'enyi-2 and von Neumann entropies in random bosonic Gaussian systems. We computed the R\'enyi-2 Page curve and Page correction when all the initial squeezing strengths are equal, and we proved various results on the typicality of the R\'enyi-2 and von Neumann measures of entanglement. Given that the R\'enyi-2 entropy is a function of only the purity, it is often tractable to measure experimentally. It would be interesting to compare the analytic formula in \cref{thm:page-curve} to an experimental Gaussian Boson Sampling device to determine how well it is generating and maintaining bipartite entanglement.

We have identified several open problems that would generalize and expand our results. One such open problem is to prove \cref{conj:entanglement-derivative}, which would allow our results on the Page curve to apply more generally. Perhaps the most important remaining task is to complete \cref{tab:typicality} by proving typicality of entanglement in the remaining regimes, such as the regime of unequal squeezing and the von Neumann entropy. 

For the latter, we note two potentially fruitful avenues. The first comes from the formula for the von Neumann entropy of a Gaussian state with covariance matrix $\sigma$ given in Ref.~\cite{kim_renyi_2018}. Let $D_\sigma \coloneqq \sqrt{\det\sigma} = \e^{S_2}$, where $S_2$ is the R\'enyi-2 entropy of $\sigma$. Ref.~\cite{kim_renyi_2018} derives expressions for all the R\'enyi-$\alpha$ entropies, including the von Neumann entropy, as functions of $D_\sigma$. In our work, we found an expression for $\Expval_U S_2(U)$ in the regime of equal squeezers by expanding $S_2(U)$ in a power series and exactly computing asymptotic expectation values over the unitary Haar measure. To find the R\'enyi-$\alpha$ entropy $\Expval_U S_\alpha(U)$, one could similarly attempt to expand in powers of $D_\sigma$, and therefore compute $\Expval_U S_\alpha(U)$ by computing $\Expval_U \e^{j S_2(U)}$ for various values of $j$. A second potential way of computing $\Expval_U S_1(U)$ is similar, where one could use the formula
\begin{equation}
    S_1 = \frac{1}{2}\log \det\bargs{\frac{\sigma+\i \Omega}{2}} + \frac{1}{2}\Tr\bargs{\operatorname{arccoth}(\i \Omega \sigma)\i \Omega \sigma},
\end{equation}
where $\Omega$ is the $2n\times 2n$ symplectic form given in \cref{eq:symplectic-form} \cite{wilde-lecture-notes,serafini_quantum_2017,hackl_bosonic_2021}. The equations resulting from using these methods with the Weingarten calculus may potentially be too difficult to simplify at first glance, as was the case in this work. However, it would be interesting to see if the presence of the $r \mapsto 1-r$ symmetry is enough, as it was in this work, to reduce the complicated Weingarten expressions to something much more tractable and simple.

In this paper, we have only considered truly Haar-random unitaries. One important question concerns how these results translate to the case where one uses random local passive Gaussian gates to generate random unitary circuits of finite depth. Indeed, an interesting open problem is to determine the depth dependence of entanglement, sampling complexity, and gate complexity in linear optical circuits. Sampling complexity refers to the classical complexity of generating samples from the output probability distribution defined by a fixed depth linear optical circuit, and gate complexity refers to the minimum number of nearest-neighbor beamsplitters required to generate the probability distribution. Currently, the precise relationship between entanglement and complexity is largely unknown.
Numerical analyses of entanglement dynamics in linear optical circuits have been reported in Refs.~\cite{zhuangScramblingComplexityPhase2019,zhouNonunitaryEntanglementDynamics2021}. Partial analytical work was done in Ref.~\cite{zhang_entanglement_2021}, but only in the regime where one or a small number of modes are not initially vacuum. On the contrary, a typical Gaussian Boson Sampling experiment initially squeezes many or all of the modes. Information on such entanglement growth could yield insights on implementations of Gaussian Boson Sampling experiments as well as the complexity of computing output probabilities from such experiments.
On the complexity side, many recent works have studied classical simulation and classical sampling complexity of linear optical circuits in certain regimes of low depth and the phase transition at which the complexity passes from easy to hard \cite{deshpandeDynamicalPhaseTransitions2018,muraleedharanQuantumComputationalSupremacy2019,ohClassicalSimulationBosonic2021,maskaraComplexityPhaseDiagram2022,ohClassicalSimulationBoson2022}. 
Indeed, both entanglement and complexity are expected grow with depth, and further study may reveal that the relationship is even more intimate.

In this work, we have characterized the entanglement properties of Gaussian states such as they arise in Gaussian Boson Sampling. In this setting, we also know that sampling from Fock basis measurements of the Gaussian state is computationally intractable. It remains an exciting question to better understand the role that entanglement plays in this context. An important aspect of this direction is to understand how entanglement and measurement bases interact. After all, some form of non-Gaussianity is crucial to generate complexity in bosonic computations \cite{chabaudResourcesBosonicQuantum2022}.

\begin{acknowledgments}
We thank Marcel Hinsche, Marios Ioannou, Kunal Sharma, and Brayden Ware for discussions related to the topic of this paper. 
JTI, AE, and AVG acknowledge support from the DoE ASCR Accelerated Research in Quantum Computing program (award No.~DE-SC0020312), DARPA SAVaNT ADVENT, AFOSR, DoE QSA, NSF QLCI (award No.~OMA-2120757), DoE ASCR Quantum Testbed Pathfinder program (award No.~DE-SC0019040), NSF PFCQC program, ARO MURI, and AFOSR MURI.
JTI thanks the Joint Quantum Institute at the University of Maryland for support through a JQI fellowship.
DH acknowledges funding from the US Department of Defense through a QuICS Hartree fellowship.
AD acknowledges funding provided by the Institute for Quantum Information and Matter, an NSF Physics Frontiers Center (NSF Grant PHY-1733907), the National Science Foundation RAISE-TAQS 1839204, and Amazon Web Services, AWS Quantum Program.
Specific product citations are for the purpose of clarification only, and are not an endorsement by the authors or NIST.
\end{acknowledgments}

\toc

\clearpage
\onecolumngrid\appendix
\renewcommand{\contentsname}{Appendices}
\tableofcontents

\section{Preliminaries}
\label{ap:prelim}

In this preliminary appendix, we will establish some notation and equations that will be used throughout the rest of the appendices. In particular, in \cref{ap:gaussian-states}, we review bosonic Gaussian states and describe our setup. In \cref{ap:weingarten}, we describe integration over the unitary group with the Weingarten calculus. Finally, in \cref{ap:equal-squeeze}, we restrict our attention to the case when all initial squeezing strengths are equal and derive a series formula for the R\'enyi-2 entropy that is used in many of our proofs.

\subsection{Bosonic Gaussian states}
\label{ap:gaussian-states}

Here, we describe the setup and fix the notation required for the proofs of our main results. We consider a very similar setup as the one described in Ref.~\cite{fukuda_typical_2019} and use much of the same notation as them. For a review of bosonic Gaussian states, we recommend Ref.~\cite{serafini_quantum_2017}. Since we are only interested in entanglement properties, the first moments --- displacements --- of the Gaussian states will be irrelevant, and we will ignore them.

We consider a system of $n$ bosonic modes. Each mode $1 \leq i \leq n$ is initially in a squeezed state with squeezing strength $s_i \in \bbR$. Define the diagonal matrix $Z = \mathrm{diag}\pargs{\e^{2 s_1}, \dots, \e^{2s_n}}$. The initial state can be represented by the covariance matrix $\sigma_0 = Z \oplus Z^{-1}$. Define $A = \frac{1}{2}(Z - Z^{-1})$ and $B = \frac{1}{2}(Z + Z^{-1})$.

The set of all passive Gaussian unitaries --- that is, energy-conserving unitaries --- acting on $n$ modes is $\Sp(2n) \cap \O(2n)$ which acts on the covariance matrix by conjugation. Here $\O(2n)$ is the orthogonal group of $2n\times 2n$ matrices, and $\Sp(2n)$ is the real symplectic group of $2n\times 2n$ matrices defined with respect to the symplectic form
\begin{equation}
    \label{eq:symplectic-form}
    \Omega \coloneqq \begin{pmatrix}0_{n\times n} & \bbI_{n\times n}\\-\bbI_{n\times n} & 0_{n\times n} \end{pmatrix}.
\end{equation}
Let $\eta\colon \U(n)\to \Sp(2n)\cap \O(2n)$ be the isomorphism
\begin{equation}
    \eta(U) = \begin{pmatrix} \Re(U) & \Im(U)\\ -\Im(U) & \Re(U) \end{pmatrix}.
\end{equation}
We evolve the initial state with covariance matrix $\sigma_0$ by a passive Gaussian unitary, which corresponds to a $U \in \U(n)$. The resulting state is $\tilde\sigma(U) \coloneqq \eta(U) \sigma_0 \eta(U^\dag) = \eta(U) \sigma_0 \eta(U)^T$.

Define the $k\times n$ matrix $P$ and the $n\times n$ projector $\Pi$ as
\begin{equation}
    P \coloneqq \begin{pmatrix} \bbI_{k\times k} & 0_{k\times (n-k)} \end{pmatrix}, \qquad
    \Pi \coloneqq \begin{pmatrix}\bbI_{k\times k} & 0_{k\times (n-k)}\\0_{(n-k)\times k} & 0_{(n-k)\times (n-k)} \end{pmatrix} = \bbI_{k\times k} \oplus 0_{(n-k)\times (n-k)}.
\end{equation}
Then let $\hat P \coloneqq P \oplus P$ and $\hat \Pi = \Pi \oplus \Pi$. The covariance matrix corresponding to the reduced state on the first $k\leq n$ modes is $\sigma(U) \coloneqq \hat P \tilde\sigma(U) \hat P^T$. Denote the element-wise complex conjugate of the unitary $U$ by $\bar U$, and the conjugate transpose by $U^\dag$. By simply doing the matrix multiplication, one finds that
\begin{equation}
    \label{eq:sigma-U}
    \sigma(U) = \frac{1}{2}\begin{pmatrix} 
        P \brackets{U B U^\dag + \bar U B \bar U^\dag + U A \bar U^\dag + \bar U A U^\dag} P^T & -\i P \brackets{U B U^\dag - \bar U B \bar U^\dag - U A \bar U^\dag + \bar U A U^\dag} P^T \\
        \i P \brackets{U B U^\dag - \bar U B \bar U^\dag + U A \bar U^\dag - \bar U A U^\dag} P^T  &  P \brackets{UBU^\dag + \bar U B \bar U^\dag - U A \bar U^\dag - \bar U A U^\dag} P^T
    \end{pmatrix}.
\end{equation}
Note that $\sigma(U)$ is a covariance matrix on $k$ modes and is correspondingly a positive $2k\times 2k$ matrix. Throughout this paper, we define $r \coloneqq k / n$. One can derive from \cref{eq:weingarten-formula} that the average over the Haar measure is $\Expval_U U B U^\dag = \Expval_U \bar U B \bar U^\dag = \frac{\Tr B}{n}\bbI_{n \times n}$, whereas all the other terms have expectation value $0$ since they do not contain an even number of $U$'s and $\bar U$'s. Therefore,
\begin{equation}
    \label{eq:sigma-U-trace}
    \Expval_{U\in\U(n)} \sigma(U) = \frac{\Tr B}{n}\bbI_{2k\times 2k}.
\end{equation}

The \textit{symplectic eigenvalues} of $\sigma(U)$ are the positive eigenvalues of $\i \Omega \sigma(U)$. There are $k$ symplectic eigenvalues labeled as $\nu_i$ for $1\leq i \leq k$. Let the von Neumann entropy of the reduced state be $S_1(U)$, and the R\'enyi-2 entropy of the reduced state be $S_2(U)$. Then $S_j(U) = \sum_{i=1}^k h_j(\nu_i)$, where $h_1(x) = \frac{x+1}{2}\log \frac{x+1}{2}-\frac{x-1}{2}\log\frac{x-1}{2}$ and $h_2(x) = \log x$ \cite{serafini_quantum_2017}. The R\'enyi-2 entropy of the reduced state takes a particularly nice form in terms of the standard eigenvalues of $\sigma(U)$, namely $S_2(U) = \frac{1}{2}\log \det \sigma(U) = \frac{1}{2}\Tr \log \sigma(U)$. The dependence of $S_1(U)$ and $S_2(U)$ on $r$, $n$, and $s$ is implicit.

\subsection{Weingarten calculus}
\label{ap:weingarten}

Since we are interested in average entanglement, we will be averaging over the unitary group $\U(n)$ with respect to the unique unit normalized Haar measure. To do so, we will use the \textit{Weingarten calculus} \cite{weingarten_asymptotic_1978,collins_moments_2002}. For a matrix $U$, let $U_{ij}$ denote the entry in row $i$ and column $j$. Then,
\begin{equation}
    \label{eq:weingarten-formula}
    \Expval_{U\in\U(n)} U_{i_1j_1}\dots U_{i_qj_q} 
    \bar U_{i'_1j'_1}\dots \bar U_{i'_qj'_q} = \sum_{\sigma,\tau \in S_q} \delta_{i_1i'_{\sigma(1)}}\dots \delta_{i_qi'_{\sigma(q)}}\delta_{j_1j'_{\tau(1)}}\dots \delta_{j_qj'_{\tau(q)}} \mathrm{Wg}(\sigma \tau^{-1}, n),
\end{equation}
where $S_q$ denotes the permutation group on $q$ elements. $\mathrm{Wg}$ is called the \textit{Weingarten function}. In our proofs, we will need the asymptotic form of the Weingarten function, which is given by
\begin{equation}
    \label{eq:asymptotic-weingarten}
    \mathrm{Wg}(\sigma, n) = \frac{1}{n^{q+\abs{\sigma}}} \prod_i (-1)^{\abs{c_i}-1}C_{\abs{c_i} - 1} + \bigO{n^{-q-\abs{\sigma} - 2}},
\end{equation}
where $\abs{\sigma}$ denotes the minimum number of transpositions needed to generate the permutation $\sigma$, $C_m=(2m)! / m!(m+1)!$ is the $m^{\rm th}$ Catalan number, and $\sigma$ is a product of cycles $c_i$ of length $\abs{c_i}$.

\subsection{Series formula for the R\'enyi-2 entropy}
\label{ap:equal-squeeze}

In this subappendix, we Taylor expand $S_2(U) = \frac{1}{2}\Tr \log \sigma(U)$ and derive a series formula for the R\'enyi-2 entropy when all the initial squeezing values are equal. Hence, for each $1 \leq i \leq n$, we set $s_i = s$. Crucially, the resulting formula is a series in the squeezing strength $s$, and the effect of the unitary $U$ is separated from that of the squeezing strength $s$.

We would like to apply the Taylor series for the matrix logarithm, and hence must first consider its convergence. We find that
\begin{align}
    \norm{\sigma(U) - \bbI_{2k \times 2k}}
    &= \norm{\hat P \eta(U) \parentheses{\sigma_0 - \bbI_{2n\times 2n}} \eta(U)^T \hat P^T}\\
    &= \norm{\hat\Pi \eta(U) \parentheses{\sigma_0 - \bbI_{2n\times 2n}} \eta(U)^T \hat \Pi^T}\\
    &\leq \norm{\hat\Pi}^2 \norm{\eta(U) \parentheses{\sigma_0 - \bbI_{2n\times 2n}} \eta(U)^T}\\
    &= \norm{\eta(U) \parentheses{\sigma_0 - \bbI_{2n\times 2n}} \eta(U)^T}\\
    &= \norm{\sigma_0 - \bbI_{2n\times 2n}} \\
    &= \max\cbargs{\abs{\e^{2s} - 1}, \abs{\e^{-2s} - 1}}\\
    &= \e^{2 \abs{s}} - 1,
\end{align}
and therefore the Taylor series for $\log$,
\begin{equation}
    \log \sigma(U) = -\sum_{j=1}^{\infty} \frac{(-1)^j}{j} \parentheses{\sigma(U) - \bbI_{2k\times 2k}}^j,
\end{equation}
converges for all $\abs{s} < R \coloneqq \frac{1}{2}\log 2$. Since $\sigma(U)$ is a positive, real symmetric matrix, this expression is indeed real and nonnegative. To make this work for all $s \in \bbR$, we can let $N \geq 1$ and consider
\begin{align}
    S_2(U)
    &= \frac{1}{2}\log\det\sigma(U)\\
    &= \frac{1}{2}\log\pargs{N^{2k} \det\frac{\sigma(U)}{N}}\\
    &= k\log N + \frac{1}{2}\Tr \log\pargs{\frac{1}{N}\sigma(U)},
\end{align}
which follows from the fact that the determinant of products is equal to the product of determinants. For any given $s\in\bbR$, we can choose $N$ large enough such that $\norm{\frac{1}{N}\sigma(U)-\bbI} < 1$ and therefore the Taylor series for $\log$ can be used. We therefore find that for large enough $N$,
\begin{equation}
    \label{eq:log-power-series}
    S_2(U) = rn \log N-\frac{1}{2}\Tr \sum_{j=1}^{\infty} \frac{(-1)^j}{j} \parentheses{\frac{1}{N}\sigma(U) - \bbI_{2k\times 2k}}^j.
\end{equation}

When all $s_i = s$, $A$ simplifies to $A = \sinh(2s)\bbI_{n\times n}$ and $B$ to $B = \cosh(2s)\bbI_{n\times n}$, and therefore $\sigma(U)$ simplifies to $\sigma(U) = \cosh(2s)\bbI_{2k\times 2k} + \sinh(2s)M$, where
\begin{equation}
    \label{eq:M-matrix}
    M \coloneqq \begin{pmatrix}
        P \Re(\bar U U^\dag) P^T & P \Im(\bar U U^\dag) P^T\\
        P \Im(\bar U U^\dag) P^T & - P \Re(\bar U U^\dag) P^T
    \end{pmatrix}.
\end{equation}
With \cref{eq:log-power-series},
\begin{align}
    S_2(U)
    &= nr\log N -\frac{1}{2}\Tr \sum_{j=1}^{\infty} \frac{(-1)^j}{j} \parentheses{\brackets{\frac{1}{N}\cosh(2s)-1}\bbI_{2k\times 2k} + \frac{1}{N}\sinh(2s)M}^j\\
    &= nr\log N-\frac{1}{2}\Tr \sum_{j=1}^{\infty} \frac{(-1)^j}{j} \sum_{\ell=0}^j \binom{j}{\ell} \frac{1}{N^\ell}\sinh^{\ell}(2s)\brackets{\frac{1}{N}\cosh(2s)-1}^{j-\ell} M^\ell\\
    \begin{split}
        &= nr\log N-rn \sum_{j=1}^\infty \frac{(-1)^j}{j} \brackets{\frac{1}{N}\cosh(2s) -1}^j\\
        &\qquad\quad-\frac{1}{2}\Tr \sum_{j=1}^{\infty} \frac{(-1)^j}{j} \sum_{\ell=1}^j \binom{j}{\ell} \frac{1}{N^\ell}\sinh^{\ell}(2s)\brackets{\frac{1}{N}\cosh(2s)-1}^{j-\ell} \Tr M^\ell
    \end{split}\\
    &= nr\log N + nr\log(\cosh(2s)/N) -\frac{1}{2} \sum_{j=1}^{\infty} \frac{(-1)^j}{j} \sum_{\ell=1}^j \binom{j}{\ell} \frac{1}{N^\ell}\sinh^{\ell}(2s)\brackets{\frac{1}{N}\cosh(2s)-1}^{j-\ell} \Tr M^\ell\\
    &= nr\log\cosh(2s) -\frac{1}{2} \sum_{\ell=1}^{\infty} \frac{1}{N^\ell}\sinh^\ell(2s)\Tr M^\ell\sum_{j=\ell}^\infty \frac{(-1)^j}{j} \binom{j}{\ell}\brackets{\frac{1}{N}\cosh(2s)-1}^{j-\ell} \\
    &= nr\log\cosh(2s) -\frac{1}{2} \sum_{\ell=1}^{\infty} \frac{1}{N^\ell}\sinh^\ell(2s)\Tr M^\ell \frac{(-1)^\ell}{\ell} \parentheses{N \operatorname{sech}(2s)}^\ell \\
    &= nr\log\cosh(2s) -\frac{1}{2} \sum_{\ell=1}^{\infty}\frac{(-1)^\ell}{\ell} \tanh^\ell(2s)\Tr M^\ell.
\end{align}
Of course, $S_2(U)$ is independent of the choice of $N$, and hence the $N$ dependence has dropped out.

The only thing left to compute is $\Tr M^\ell$. Since we are now only dealing with traces, we can replace $P$ with $\Pi$ in $M$. This nicely simplifies some formulas, since $\Pi$ is a square matrix, and indeed a projector, whereas $P$ is a rectangular matrix. Henceforth, we will therefore let
\begin{equation}
    M = \begin{pmatrix}
        \Pi R \Pi & \Pi I \Pi\\
        \Pi I \Pi & - \Pi R \Pi
    \end{pmatrix},
\end{equation}
where $R = \Re (\bar U U^\dag)$ and $I = \Im (\bar U U^\dag)$. Then, doing the matrix multiplication,
\begin{align}
    M^2 
    &= \begin{pmatrix}
        \Pi R\Pi R \Pi + \Pi I \Pi I \Pi & \Pi R \Pi I \Pi - \Pi I \Pi R \Pi\\
        \Pi I \Pi R \Pi - \Pi R \Pi I \Pi & \Pi R\Pi R \Pi + \Pi I \Pi I \Pi
    \end{pmatrix}.
\end{align}
We then notice that $M^2 = \begin{pmatrix} 
D&E\\-E&D
\end{pmatrix}$, where $D = \Re (W)$, $E = \Im(W)$, and $W = \Pi U \bar U^\dag \Pi \bar U U^\dag \Pi$. It is then easy to verify that $M^{2j} = \begin{pmatrix} 
D_j&E_j\\-E_j&D_j
\end{pmatrix}$, where $D_{j+1} = D_j D - E_j E$, $E_{j+1} = D_j E + E_j D$, $D_1 = D$, and $E_1 = E$. It is then also easy to verify that this recurrence relation is solved by $D_j = \Re(W^j)$ and $E_j = \Im(W^j)$. Finally, $\Tr M^{2j} = 2 \Tr D_j$. 

Using this recurrence, we can do the matrix multiplication $M^{2j+1} = M^{2j}M$ to find that $\Tr M^{2j+1} = 0$. Hence, we only need to worry about even powers of $M$, giving
\begin{equation}
    S_2(U) = nr\log\cosh(2s) -\frac{1}{2} \sum_{\ell=1}^{\infty} \frac{1}{2\ell} \tanh^{2\ell}(2s)\Tr M^{2\ell}.
\end{equation}
We then use that $\Tr M^{2j} = 2 \Tr D_j$ to find that
\begin{equation}
    S_2(U) = nr\log\cosh(2s) - \sum_{\ell=1}^{\infty} \frac{1}{2\ell} \tanh^{2\ell}(2s) \Tr \Re W^\ell.
\end{equation}
Furthermore, since $W$ is Hermitian, $\Tr W^j = \Tr \bar W^j$. Therefore, we arrive at
\begin{align}
    S_2(U) &= nr\log\cosh(2s) - \sum_{\ell=1}^\infty \frac{1}{2\ell} \tanh^{2\ell}(2s)\Tr W^\ell\label{eq:renyi-2}\\
    &= n\sum_{\ell=1}^\infty \frac{\tanh^{2\ell}(2s)}{2\ell} \parentheses{r - \frac{1}{n} \Tr W^\ell},\label{eq:simplified-s2}
\end{align}
where in the last step we used the Taylor expansion of $\log\cosh(\operatorname{arctanh} t)$ in the variable $t = \tanh(2s)$.

\cref{eq:renyi-2,eq:simplified-s2} hint at why equal initial squeezings simplify the problem of studying averaged entanglement properties. Specifically, the contribution from the squeezing strength and the contribution from the unitary are separated. Thus, to determine averaged entanglement properties, we only need to deal with the matrix $W = \Pi UU^T \Pi \bar U \bar U^T \Pi$.

\section{R\'enyi-2 and von Neumann entropies --- Proof of \texorpdfstring{\cref{prop:s1-bound}}{\autoref*{prop:s1-bound}}}
\label{ap:entropy-bounds}

In this appendix, we prove \cref{prop:s1-bound}. We derive the maximum of the R\'enyi-2 and von Neumann entropies, and the former will be of use later when we derive the R\'enyi-2 Page correction. Furthermore, we prove that the von Neumann entropy can be bounded by the R\'enyi-2 entropy. In this way, our results on the R\'enyi-2 Page curve can be used to bound the von Neumann Page curve.

We begin by proving that the maximum of the R\'enyi-2 entropy is
\begin{equation}
    \label{eq:max-s2}
    \max_{U \in \U(n)} S_2(U) = n\min(r,1-r) \log\cosh(2s).
\end{equation}
\begin{proof}
Recall that $W = F^\dag F$, where $F = \Pi \bar U U^\dag \Pi$. Therefore, $W$ is a nonnegative operator, and $\Tr W^\ell \geq 0$. The proposition then immediately follows from \cref{eq:renyi-2} if we can show that for every $r \leq 1/2$, there exists a unitary such that $W = 0$. The case when $r > 1/2$ is taken care of by the fact that the R\'enyi-2 entropy is symmetric under $r \mapsto 1-r$ since the global state on the $n$ modes is pure. Therefore, we now assume that $r\leq 1/2$, and we show that there are unitaries $U \in \U(n)$ such that $W = 0$.

Since $F^\dag F$ is nonnegative, we need to prove that there exists a unitary such that $F = \Pi \bar U U^\dag \Pi = 0$. Hence we must prove that there exists a $U$ such that $(\bar U U^\dag)_{ji} = 0$ for all $1 \leq i ,j \leq k = rn$. Equivalently, we can conjugate the expression, giving $\overline{(\bar U U^\dag)}_{ji}= \sum_{a=1}^n U_{ia} U_{ja} = 0$. Therefore, we just need to find a set of $k = rn$ orthonormal vectors $\calS = \set{\ket{\psi_1},\dots \ket{\psi_k}}$ in $\bbC^n$ such that $\braket{\bar \psi_i \vert \psi_j} = 0$. Let $\ket{\psi_i} = \ket{R_i} + \i \ket{I_i}$ for real vectors $\ket{R_i}$ and $\ket{I_i}$. Since $\calS$ is orthonormal, we find that
\begin{align}
    \delta_{ij}
    &= \braket{\psi_i \vert \psi_j}\\
    &= (\bra{R_i} - \i \bra{I_i})(\ket{R_j} + \i \ket{I_j})\\
    &= \braket{R_i \vert R_j} + \braket{I_i \vert I_j}  + \i \parentheses{\braket{R_i \vert I_j} - \braket{I_i \vert R_j}}.
\end{align}
The condition that $\braket{\bar \psi_i \vert \psi_j} = 0$ implies that
\begin{align}
    0 &= \braket{\bar \psi_i \vert \psi_j}\\
    &= (\bra{R_i} + \i \bra{I_i})(\ket{R_j} + \i \ket{I_j})\\
    &= \braket{R_i \vert R_j} - \braket{I_i \vert I_j}  + \i \parentheses{\braket{R_i \vert I_j} + \braket{I_i \vert R_j}}.
\end{align}
Hence, we just need to find $k$ vectors $\ket{R_i} \in \bbR^n$ and $k$ vectors $\ket{I_i} \in \bbR^n$ satisfying
\begin{enumerate}
    \item $\braket{R_i \vert R_j}+\braket{I_i \vert I_j}=\delta_{ij}$,
    \item $\braket{R_i \vert R_j}-\braket{I_i \vert I_j}= 0$,
    \item $\braket{R_i \vert I_j} + \braket{I_i \vert R_j} = 0$, and
    \item $\braket{R_i \vert I_j} - \braket{I_i \vert R_j} = 0$,
\end{enumerate}
for all $i,j \in \set{1,\dots, k}$. This is trivial since $ k \leq n/2$, and we give one construction here. Let $\calR = \set{\ket{R_1}, \dots , \ket{R_k}}$ be an orthogonal set satisfying $\braket{R_i \vert R_j} = \delta_{ij}/2$. Since $\mathrm{dim}[\mathrm{span}(\calR)^\perp] = n-k \geq k$, we can choose another orthogonal set $\calI = \set{\ket{I_1}, \dots, \ket{I_k}}$ satisfying $\braket{I_i \vert I_j} = \delta_{ij}/2$ such that $\mathrm{span}(\calI) \cap \mathrm{span}(\calR) = \set{0}$. This immediately means that condition 3 and 4 are satisfied, because $\braket{R_i \vert I_j} = 0$ for all $i$ and $j$. Furthermore, condition 1 is satisfied because $\braket{R_i \vert R_j} + \braket{I_i \vert I_j} = \delta_{ij}/2 + \delta_{ij}/2 = \delta_{ij}$. Finally, condition 2 is satisfied because $\braket{R_i \vert R_j} - \braket{I_i \vert I_j} = \delta_{ij}/2 - \delta_{ij}/2 = 0$.
\end{proof}

\medskip

Next, we use \cref{eq:max-s2} to prove that the maximum of the von Neumann entropy is
\begin{equation}
    \max_{U\in\U(n)}S_1(U) = n\min(r,1-r) h_1(\cosh(2s)).
\end{equation}
\begin{proof}
From Ref.~\cite{kim_renyi_2018}, the R\'enyi-$\alpha$ entropies for $\alpha \geq 1$ are all increasing with increasing $\det\sigma$. Therefore, when $S_2$ increases, so does $S_\alpha$. It follows that the unitary that maximizes $S_2$ also maximizes $S_1$. Then, when $r\leq 1/2$, \cref{eq:max-s2,eq:M-matrix} together imply that $S_1$ is maximized when $\sigma = \cosh(2s)\bbI_{2k\times 2k}$ and thus when $\nu_i=\cosh(2s)$ for each $i$, yielding $\max_U S_1(U) = nr h_1(\cosh(2s))$. The case when $r>1/2$ is taken care of by the symmetry of the von Neumann entropy when $r\mapsto 1-r$ since we are dealing with a bipartite system \cite{nielsen_quantum_2010}.
\end{proof}

\medskip

Finally, we prove the bound
\begin{equation}
    S_1(U) < S_2(U) + n\min(r,1-r)(1-\log 2),
\end{equation}
which was first derived in Ref.~\cite[Eq.~15]{adesso_strong_2016}.
\begin{proof}
Recall that $S_j = \sum_{i=1}^k h_j(\nu_i)$, where $\nu_i\geq 1$ are the symplectic eigenvalues of the covariance matrix $\sigma$, $h_2(x) = \log x$, and $h_1(x) = \frac{x+1}{2}\log \frac{x+1}{2}-\frac{x-1}{2}\log \frac{x-1}{2}$ (see \cref{ap:gaussian-states}). We will prove that $h_1(x) < h_2(x) + 1-\log 2$, which proves the claim. We do this by noting that $h_1(x) - h_2(x)$ is monotonically increasing, so that for any $x$, $h_1(x) - h_2(x) < \lim_{y\to\infty} (h_1(y) - h_2(y))$. A simple calculation shows that this limit is equal to $1-\log 2$.
\end{proof}

\section{R\'enyi-2 Page curve}

In this appendix, we prove our main results on the R\'enyi-2 Page curve. In \cref{ap:page}, we prove that asymptotically in $n$, $\Expval_{U\in \U(n)} S_2(U) = n \alpha(s,r)-\lambda(s,r)+\littleo{1}$. We then derive the exact formula for the linear term $\alpha(s,r)$. Then, in \cref{ap:special-page}, we simplify $\alpha(s,1/2)$ and use the result to prove a simple expression for the R\'enyi-2 Page correction at $r=1/2$. Finally, in \cref{ap:const}, we derive a formula for the constant term $\lambda(s,r)$ up to constant factors $a^{(\ell)}$. We then derive the formula for $a^{(\ell)}$, simplify it for $\ell \in \set{1,2,3,4,5}$, and conjecture the exact value of $a^{(\ell)}$ for all $\ell$.

\subsection{Linear term --- Proof of \texorpdfstring{\cref{thm:page-curve}}{\autoref*{thm:page-curve}}}
\label{ap:page}

\cref{thm:page-curve} concerns the expectation value of $S_2(U)$ over $\U(n)$ when all the initial squeezing values are equal. From \cref{eq:renyi-2}, we see that it only remains to compute $\Expval_U\Tr W^\ell$. Writing the matrix multiplication of $W^\ell = (\Pi U U^T \Pi \bar U \bar U^T \Pi)^\ell$ in terms of the matrix entries of $U$ and $\Pi$ and simplifying, we find that
\begin{equation}
\begin{aligned}
    \Tr W^\ell
    &= \sum_{i_1,\dots,i_{2\ell}=1}^k\sum_{i_1',\dots,i_{2\ell}'=1}^k\sum_{j_1,\dots,j_{2\ell}=1}^n\sum_{j_1',\dots,j_{2\ell}'=1}^n \\ 
    &\qquad \delta_{i_{2\ell}',i_1}\delta_{i_1',i_2}\delta_{i_2',i_3}\delta_{i_3',i_4} \dots \delta_{i_{2\ell-1}',i_{2\ell}}\\
    &\qquad \times \delta_{j_1,j_2}\delta_{j_1',j_2'}\delta_{j_3,j_4}\delta_{j_3',j_4'} \dots \delta_{j_{2\ell-1},j_{2\ell}}\delta_{j_{2\ell-1}',j_{2\ell}'} \\
    &\qquad \times U_{i_1,j_1}\dots U_{i_{2\ell},j_{2\ell}} \bar U_{i_1',j_1'}\dots \bar U_{i_{2\ell}',j_{2\ell}'}.
\end{aligned}
\end{equation}
We note that this is a simple result of matrix multiplication. The restriction on the $i$ and $i'$ indices to $\set{1,\dots,k}$ is a result of the fact that $\Pi_{a,a} = 0$ for all $a > k$.
Applying \cref{eq:weingarten-formula}, we immediately find that
\begin{equation}
\label{eq:weingarten-trw}
\begin{aligned}
    \Expval_{U \in \U(n)} \Tr W^\ell
    &= \sum_{i_1,\dots,i_{2\ell}=1}^k \sum_{i_1',\dots,i_{2\ell}'=1}^k \sum_{j_1,\dots,j_{2\ell}=1}^n \sum_{j_1',\dots,j_{2\ell}'=1}^n \sum_{\sigma,\tau \in S_{2\ell}} \text{Wg}(\sigma\tau^{-1},n)\\
    &\qquad \times \delta_{i_{2\ell}', i_1} \delta_{i_1',i_2}\delta_{i_2',i_3}\dots \delta_{i_{2\ell-1}', i_{2\ell}} \\
    &\qquad \times \delta_{j_1,j_2}\delta_{j_1',j_2'}\dots \delta_{j_{2\ell-1},j_{2\ell}}\delta_{j_{2\ell-1}',j_{2\ell}'}\\
    &\qquad \times \delta_{i_1,i_{\sigma(1)}'} \dots \delta_{i_{2\ell},i_{\sigma(2\ell)}'}\\
    &\qquad \times \delta_{j_1,j_{\tau(1)}'} \dots \delta_{j_{2\ell},j_{\tau(2\ell)}'}.
\end{aligned}
\end{equation}
Simplifying \cref{eq:weingarten-trw} at first seems impossible, but it will turn out that we do not need to. All we need to learn from it is the following lemma.

\medskip

\begin{lemma}
\label{lem:form-of-w}
Fix a positive integer $\ell$. There exist coefficients $\alpha_{d}^{(\ell)}$ for $d \in \set{\ell+1, \ell+2, \dots, 2\ell}$ such that
\begin{equation}
    f_\ell(r) \coloneqq \lim_{n\to\infty} \Expval_{U \in \U(n)} \frac{1}{n} \Tr W^\ell = \sum_{d=\ell+1}^{2\ell} \alpha_d^{(\ell)} r^d.
\end{equation}
\end{lemma}
\begin{proof}
The proof will proceed as follows. First, we will prove that $\Tr W^\ell$ contains a term proportional to $n$ and no terms proportional to $n^a$ for any $a > 1$. Therefore, $f_\ell(r)$ is indeed independent of $n$. Next, we will prove that $f_\ell(r)$ has no terms $r^a$ for $a > 2\ell$ and $a\leq \ell$. Throughout this proof, we interpret the delta functions in \cref{eq:weingarten-trw} as constraints on the summations. Different permutations on the indices result in a different number of constraints and hence terms with different powers of $n$ and $k$.
 
Recall that $W = F^\dag F$ where $F = \Pi \bar U U^\dag \Pi$. Therefore, $\Tr W$ is equal to the Frobenius norm $\norm{F}_F^2$, which is the sum of the square absolute values of the entries of $F$. Thus, by removing the projector $\Pi$ from $W$ (i.e.~setting it to $\bbI$), the trace cannot decrease. It follows that the presence of $\Pi$ cannot increase the trace $\Tr W^\ell$. Getting rid of the $\Pi$ from $W$ and using the cyclic nature of the trace, we see that $\Tr W^\ell \leq n$. Furthermore, consider \cref{eq:weingarten-trw} with $\sigma = \tau$ defined by $\sigma(1) = 2\ell$, $\sigma(2) = 1$, $\sigma(3) = 2$, $\dots$, $\sigma(2\ell) = 2\ell-1$. Then, $\sigma \tau^{-1}$ is the identity, $\text{Wg}(\sigma \tau^{-1}, n)$ contributes a factor of $n^{-2\ell}$. With this $\sigma$, the sum over the $i$ and $i'$ yields a factor of $k^{2\ell}$. Finally, with this chosen $\tau$, the sum over $j$ will yield $\delta_{j_1',j_2'}\dots \delta_{j_{2\ell-1}',j_{2\ell}'} \delta_{j_1',j_{2\ell}'} \delta_{j_2',j_1'}\dots \delta_{j_{2\ell}',j_{2\ell-1}'}$. Then summing over $j'$, we get a single factor of $n$. Hence, the term with the specific permutation described above yields a term of the form $n k^{2\ell}n^{-2\ell} = n r^{2\ell}$.

We have shown that there is a term proportional to $n$ and that there are no terms proportional to $n^2$, $n^3$, etc. Since we are working asymptotically in $n$, we can therefore ignore all terms proportional to $\frac{1}{n^a}$ for every nonnegative $a$. This proves that $\lim_{n\to\infty} \Expval_{U \in \U(n)} \frac{1}{n} \Tr W^\ell$ is independent of $n$ and only depends on $r$, which justifies the definition of the function $f_\ell(r)$. The only thing left to show is that $f_\ell(r)$ has only terms $r^{\ell+1}$ through $r^{2\ell}$. So we only need to show that there are no terms $r^a$ for $a > 2\ell$ and $a \leq \ell$. We begin with the former.

To look at powers of $r$, it is sufficient to look at the powers of $k$. We therefore restrict our attention to the sum over $i$ and $i'$ in \cref{eq:weingarten-trw}. The sum over $i_1$, $\sum_{i_1=1}^k$, will give either a factor of $1$ or a factor of $k$ depending on how the index $i_1$ is constrained by the Kronecker delta functions, and similarly for $i_2$ through $i_{2\ell}$. In order to get the highest power of $k$, we require the fewest constraints on $i$ and $i'$ (i.e.~the fewest \textit{distinct} Kronecker deltas). Hence, we require $\sigma$ to be the permutation satisfying
\begin{equation}
    \delta_{i_{2\ell}', i_1} \delta_{i_1',i_2}\delta_{i_2',i_3}\dots \delta_{i_{2\ell-1}', i_{2\ell}} = \delta_{i_1,i_{\sigma(1)}'} \dots \delta_{i_{2\ell},i_{\sigma(2\ell)}'}.
\end{equation}
This permutation is $\sigma(1) = 2\ell$, $\sigma(2) = 1$, $\sigma(3) = 2$, $\dots$, $\sigma(2\ell) = 2\ell-1$. With this $\sigma$, we see that a sum over $i$ and $i'$ will give a factor of $k^{2\ell}$. Hence, $2\ell$ is the highest power of $k$ that can be achieved.

Next we need to show that $\ell+1$ is the lowest power of $k$ that can be achieved. The sum over $j$ and $j'$ can give at most a factor of $n^{\ell}$. This is because the first line of delta functions, $\delta_{j_1,j_2}\delta_{j_1',j_2'}\dots \delta_{j_{2\ell-1},j_{2\ell}}\delta_{j_{2\ell-1}',j_{2\ell}'}$, reduces the sum over $2\ell$ indices $j$ and $2\ell$ indices $j'$ down to just a sum over $\ell$ indices $j$ and $\ell$ indices $j'$. The second line of delta functions, $\delta_{j_1,j_{\tau(1)}'} \dots \delta_{j_{2\ell},j_{\tau(2\ell)}'}$, cannot be made equivalent to the first line by any choice of $\tau$; in fact, the second line imposes all new constraints. Therefore, this line further reduces the sum over $\ell$ indices $j$ and $\ell$ indices $j'$ down to just a sum over $\ell$ indices $j$ (or $\ell$ indices $j'$, but not both). Hence, the highest power of $n$ that we get from the summations over $j$ and $j'$ is $n^\ell$. Putting this together with the fact that asymptotically $\text{Wg}(\pi, n)$ is at most $n^{-2\ell}$, we find that any term coming from \cref{eq:weingarten-trw} is at most $n^{-\ell} \times \parentheses{\text{dependence on } k}$. Therefore, any powers of $k$ that are less than $\ell+1$ can be ignored; if the sum over $i$ and $i'$ yields a term that is $k^a$ for some $a \leq \ell$, then that term will be constant or decreasing with $n$. But from above, we already have terms that are proportional to $n$ resulting from \cref{eq:weingarten-trw}, and so terms that are constant or decreasing can be ignored.
\end{proof}

\medskip

As alluded to in the proof of the lemma, the asymptotic form of $\Expval_U \Tr W^\ell$ will be a term linear in $n$ times a function of $r$, plus a term constant in $n$ times a function of $r$, plus terms that decay to zero asymptotically with $n$. Hence, $\Expval_U S_2(U) = n\alpha(s,r) -\lambda(s,r) + \littleo{1}$. In this section, we are interested only in the linear term $\alpha(s,r)$ because we are computing $\Expval_U S_2(U) / n$. However, in the next section, we prove \cref{prop:const} which provides the form of $\lambda(s,r)$.

Actually computing $\alpha_d^{(\ell)}$ from \cref{eq:weingarten-trw} seems challenging. However, there is a nice workaround that uses what we know about the R\'enyi-2 entropy being symmetric under $r \mapsto 1-r$, as the total state on the $n$ modes is pure.

\medskip

\begin{lemma}
\label{lem:proof-of-alpha}
Let $\alpha_d^{(\ell)}$ be as in \cref{lem:form-of-w}. Then
\begin{equation}
    \label{eq:alpha-analytic}
    \alpha_{d}^{(\ell)} = 2 (-1)^{d-\ell-1} \binom{2 \ell-1}{\ell-1} \binom{\ell}{d-\ell-1} \frac{2\ell-d+1}{(d-1) d}.
\end{equation}
Note that this corresponds to the sequence \href{https://oeis.org/A062991}{A062991} on OEIS \cite{oeis}.
\end{lemma}
\begin{proof}
From \cref{lem:form-of-w}, asymptotically $\Expval_U \Tr W^\ell = nf_l(r)$ for a polynomial $f_\ell$ that is a sum over terms of degree $\ell+1$ through $2\ell$ in $r$. From \cref{eq:simplified-s2}, we know that
\begin{equation}
    \lim_{n\to\infty}\Expval_U \frac{1}{n} S_2(U) = \sum_{\ell=1}^\infty \frac{t^{2\ell}}{2\ell}  \parentheses{ r - f_\ell(r) },\\
\end{equation}
where $-1 < t = \tanh(2s) < 1$. This whole function must be symmetric under $r\to 1-r$ \textit{at every order in $t$}. The reason it must be symmetric at every order is as follows. Suppose we choose $s$ such that $t \in \bigO{n^{-1/2}}$. Then $\Expval_U S_2(U) = \frac{1}{2}nt^2 (r-f_1(r)) + \littleo{1}$. Hence, $r-f_1(r)$ must be symmetric since $S_2$ is symmetric. We then choose $s$ such that $t \in \bigO{n^{-1/4}}$ and consider $\Expval_U \bargs{S_2(U) - \frac{1}{2}nt^2 (r-f_1(r))}$, which will be equal to $\frac{1}{4}nt^4 (r-f_2(r)) + \littleo{1}$. Since $S_2$ is symmetric and $r-f_1(r)$ is symmetric, it follows that $r-f_2(r)$ is symmetric. We then continue inductively like this to prove that $r-f_\ell(r)$ must be symmetric for all $\ell$.

We therefore find that for all $\ell$, $f_\ell(r)$ must satisfy $r-f_\ell(r) = 1-r-f_\ell(1-r)$, or
\begin{equation}
    \label{eq:f-constraint}
    f_\ell(r) - f_\ell(1-r) = 2r-1.
\end{equation}
We then plug in the form of $f_\ell(r) = \sum_{d=\ell+1}^{2\ell} \alpha_d^{(\ell)} r^d$. \cref{eq:f-constraint} must hold for every $r$, and therefore we can equate the coefficients in front of each $r^d$ term to get a system of equations that can be solved for the values of $\alpha_d^{(\ell)}$. We then find that
\begin{align}
    2r-1
    &= \sum_{d=\ell+1}^{2\ell}\alpha_d^{(\ell)} \parentheses{r^d - (1-r)^{d}}\\
    &= \sum_{d=\ell+1}^{2\ell}\alpha_d^{(\ell)} r^d - \sum_{d=\ell+1}^{2\ell}\alpha_d^{(\ell)} \sum_{j=0}^d \binom{d}{j}(-r)^j\\
    &= \sum_{d=\ell+1}^{2\ell}\alpha_d^{(\ell)} r^d - \sum_{j=0}^{2\ell}(-r)^j \sum_{d=\max(j,\ell+1)}^{2\ell}\alpha_d^{(\ell)} \binom{d}{j}\\
    &= \sum_{d=\ell+1}^{2\ell}\alpha_d^{(\ell)} r^d - \sum_{d=0}^{2\ell} (-r)^d \sum_{j=\max(d,\ell+1)}^{2\ell}\alpha_j^{(\ell)} \binom{j}{d}\\
    &= -\sum_{d=0}^\ell (-r)^d \sum_{j=\ell+1}^{2\ell} \alpha_j^{(\ell)} \binom{j}{d} + \sum_{d=\ell+1}^{2\ell} r^d \parentheses{\alpha_d^{(\ell)} - (-1)^d \sum_{j=d}^{2\ell} \alpha_j^{(\ell)} \binom{j}{d} }.
\end{align}
Equating the degrees in $r$, we see the following conditions;
\begin{enumerate}
    \item $1 = \sum_{j=\ell+1}^{2\ell}\alpha_j^{(\ell)}$,
    \item $2 = \sum_{j=\ell+1}^{2\ell} j \alpha_j^{(\ell)}$,
    \item For $2 \leq d \leq \ell$, $0 = \sum_{j=\ell+1}^{2\ell} \binom{j}{d}\alpha_j^{(\ell)}$,
    \item For $\ell+1 \leq d \leq 2\ell$, $\alpha_d^{(\ell)} = (-1)^d \sum_{j=d}^{2\ell} \binom{j}{d} \alpha_j^{(\ell)} $.
\end{enumerate}
Note that one can derive equivalent conditions by requiring that the polynomial $(x+1/2)-f_\ell(x+1/2)$ is even in $x$. Condition 2 and condition 3 together constitute a linear system of $\ell$ linearly independent equations. To verify this, one must show that $\det C \neq 0$ where $C$ is the $\ell \times \ell$ matrix with entries $C_{ij} = \binom{\ell + i}{j}$. As noted by Benoit Cloitre in OEIS sequence \href{https://oeis.org/A000984}{A000984} \cite{oeis}, $\det C = \binom{2\ell}{\ell}$. For a proof, see Ref.~\cite{delanoy_determinant_2017}.

Since we have $\ell$ linearly independent equations for $\ell$ variables $\alpha$, if there is a solution to the four conditions then there is a \textit{unique} solution. One can then verify that a solution to the four conditions, and therefore \textit{the} solution, is given by \cref{eq:alpha-analytic}. 

When verifying the four conditions, one finds that the right hand side of all four of the conditions can be simplified in terms of the hypergeometric function $\, _2F_1$ defined as
\begin{equation}
    \label{eq:hypergeometric-function}
    \, _2F_1(-m,b;c;r) = \sum_{a=0}^m (-1)^a \binom{m}{a} \frac{(c-1)!(a+b-1)!}{(b-1)!(a+c-1)!} r^a
\end{equation}
when $m$ is nonnegative \cite{petkovsek1996,bailey_generalized_1964,zudilin_hypergeometric_2019,slater_generalized_1966,wolfram_hypergeometric2f1}. For example, condition 4, written as $(-1)^d = \sum_{j=d}^{2\ell} \binom{j}{d}(\alpha_j^{(\ell)}/\alpha_d^{(\ell)})$, reduces to $(-1)^d = \, _2F_1(d-2 \ell,d-1;d-\ell;1)$. Define the Pochhammer symbol as $(x)_a \coloneqq \frac{\Gamma(x+a)}{\Gamma(x)}$. When $a \in \bbN$, $\, _2F_1(-a,b;c;1) = \frac{(c-b)_a}{(c)_a}$ \cite{petkovsek1996,bailey_generalized_1964,zudilin_hypergeometric_2019,slater_generalized_1966,wolfram_hypergeometric2f1}. We therefore find that
\begin{align}
    \, _2F_1(d-2 \ell,d-1;d-\ell;1)
    &= \frac{(1-\ell)_{2\ell-d}}{(d-\ell)_{2\ell-d}}\\
    &= \frac{\Gamma(d-\ell)\Gamma(1-d+\ell)}{\Gamma(\ell)\Gamma(1-\ell)}\\
    &= \frac{\Gamma(d-\ell)\Gamma(1-\ell)(-d+\ell)(-d+\ell-1)\dots (1-\ell)}{\Gamma(\ell)\Gamma(1-\ell)}\\
    &= (-1)^{1+(1-\ell)-(-d+\ell)}\frac{\Gamma(d-\ell)(d-\ell)(d-\ell+1)\dots (\ell-1)}{\Gamma(\ell)}\\
    &= (-1)^d \frac{\Gamma(d-\ell) (\ell-1)!}{\Gamma(\ell) (d-\ell-1)!}\\
    &= (-1)^d,
\end{align}
hence proving condition 4. Conditions 1, 2, and 3 are similar.
\end{proof}

\medskip

Plugging in this result, we find that
\begin{align}
    \label{eq:analytic-W}
    f_\ell(r) \coloneqq \frac{1}{n}\Expval_U \Tr W^\ell &= 2\binom{2 \ell-1}{\ell-1}\sum_{i=0}^{\ell-1} (-1)^i r^{i+\ell+1} \binom{\ell}{i} \frac{\ell-i}{(\ell+i) (\ell+i+1)}\\
    &= r^{\ell+1}C_\ell \, _2F_1(1-\ell,\ell;\ell+2;r),
\end{align}
where $\, _2 F_1$ is the hypergeometric function and $C_\ell \coloneqq \frac{1}{\ell+1}\binom{2\ell}{\ell}$ is the $\ell^{\rm th}$ Catalan number. The simplification from the first to second line follows from the definition of the hypergeometric function given in \cref{eq:hypergeometric-function}. Plugging this exact formula for $f_\ell(r) = \lim_{n\to\infty}\Expval_U \frac{1}{n} \Tr W^\ell$ into \cref{eq:renyi-2}, and swapping sums, we arrive precisely at \cref{thm:page-curve}. We have therefore completed the proof.

To get a sense for $f_\ell$, we list a few here.
\begin{align*}
f_1(r) &= r^2\\
f_2(r) &= -r^4 + 2 r^3\\
f_3(r) &= 2 r^6-6 r^5+5 r^4\\
f_4(r) &= -5 r^8+20 r^7-28 r^6+14 r^5\\
f_5(r) &= 14 r^{10}-70 r^9+135 r^8-120 r^7+42 r^6\\
f_6(r) &= -42 r^{12}+252 r^{11}-616 r^{10}+770 r^9-495 r^8+132 r^7\\
f_7(r) &= 132 r^{14}-924 r^{13}+2730 r^{12}-4368 r^{11}+4004 r^{10}-2002 r^9+429 r^8\\
f_8(r) &= -429 r^{16}+3432 r^{15}-11880 r^{14}+23100 r^{13}-27300 r^{12}+19656 r^{11}-8008 r^{10}+1430 r^9.
\end{align*}
Recall that we required $G_\ell(r) \coloneqq r-f_\ell(r)$ to be symmetric under $r \mapsto 1-r$. We show some of the resulting plots in \cref{fig:Glr}. Intuitively, one can understand $G_\ell(r)$ as being an approximation to the function $m(r) \coloneqq \min(r,1-r)$; as $\ell$ increases, this approximation gets better and better.

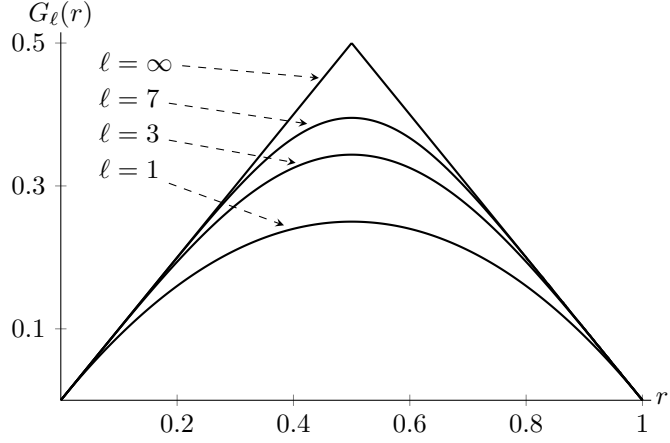
\begin{figure}
    \centering
    
    \begin{tikzpicture}
    \begin{axis}[
        width=.55\textwidth,
        height=.38\textwidth,
        xlabel={$r$},
        title={$G_\ell(r)$},
        every axis title/.style={
            at={(0,1.05)}
        },
        samples=500,
        domain=0:1,
        xtick={.2,.4,.6,.8,1},
        ytick={.1,.3,.5},
        ymax=.515, xmax=1.015,
        restrict y to domain =0:1,
        axis y line=middle,
        axis x line=middle,
        x axis line style=-,
        y axis line style=-,
        x label style={xshift=10,yshift=-5}
    ]
    
    \addplot[black, thick] plot (\x, {min(\x,1-\x)});
    
    \addplot[black, thick] plot (\x, {\x - 429 *pow(\x,8) + 2002 *pow(\x,9) - 4004 *pow(\x,10) + 4368 *pow(\x,11) - 2730 *pow(\x,12) + 924 *pow(\x,13) - 132 *pow(\x,14)});
    
    \addplot[black, thick] plot (\x, {\x - (2*pow(\x,6)-6*pow(\x,5)+5*pow(\x,4))});
    
    \addplot[black, thick] plot (\x, {\x - pow(\x,2)});

    \node (inf-label) at (axis cs:.05,.5) [anchor=north west] {$\ell=\infty$};
    \node (inf-curve) at (axis cs:.46,.45){};
    \draw[->,>=stealth,dashed](inf-label)--(inf-curve);
    
    \node (seven-label) at (axis cs:.05,.45) [anchor=north west] {$\ell=7$};
    \node (seven-curve) at (axis cs:.44,.385){};
    \draw[->,>=stealth,dashed](seven-label)--(seven-curve);
    
    \node (three-label) at (axis cs:.05,.4) [anchor=north west] {$\ell=3$};
    \node (three-curve) at (axis cs:.42,.33){};
    \draw[->,>=stealth,dashed](three-label)--(three-curve);
    
    \node (one-label) at (axis cs:.05,.35) [anchor=north west] {$\ell=1$};
    \node (one-curve) at (axis cs:.4,.24){};
    \draw[->,>=stealth,dashed](one-label)--(one-curve);
    
    \end{axis}
    \end{tikzpicture}

    \caption{The plots of $G_\ell(r) \coloneqq r - f_\ell(r)$ for various values of $\ell$. $f_\ell(r)$ is given in \cref{eq:analytic-W}. In our proof, we crucially used that $G_\ell(r)$ is symmetric under $r \mapsto 1-r$.}
    \label{fig:Glr}
\end{figure}

More specifically, we can interpret the polynomials $G_\ell(r)$ in multiple ways. Once the form of $f_\ell(r)$ is fixed by \cref{lem:form-of-w}, we showed in \cref{lem:proof-of-alpha} that $G_\ell$ is uniquely determined by the symmetry requirement that $G_\ell(r) = G_\ell(1-r)$. In the proof of \cref{lem:proof-of-alpha}, we showed that $G_\ell$ is uniquely determined by conditions 1, 2, and 3. Each of these conditions has a simple interpretation. Condition 1, $1 = \sum_{j=\ell+1}^{2\ell} \alpha_j^{(\ell)}$, is the condition that $G_\ell(1) = 0$. Condition 2, $2 = \sum_{j=\ell+1}^{2\ell} j \alpha_j^{(\ell)}$, is the condition that the derivative $G_\ell'(1) = -1$. Condition 3, $0 = \sum_{j=\ell+1}^{2\ell} \binom{j}{d}\alpha_j^{(\ell)}$ for each $2 \leq d \leq \ell$, is the condition that the $d^{\rm th}$ derivative $G_\ell^{(d)}(1) = 0$. Hence, condition 1, 2, and 3 are imposing that $G^{(d)}_\ell(0) = m^{(d)}(0)$ and $G^{(d)}_\ell(1) = m^{(d)}(1)$ for all $0 \leq d \leq \ell$. Indeed, from the derived form of $f_\ell$, we find that $G_\ell(r)$ is an approximation from below to $m(r)$, and it is an especially good approximation near the endpoints $r=0$ and $r=1$. We show some examples of this in \cref{fig:Glr}. 

We will call $G_\ell(r)$ the $\ell^{\rm th}$ order approximation to $m(r)$. From \cref{eq:simplified-s2}, $\lim_{n\to\infty}\Expval_U \frac{1}{n}S_2(U) = \sum_{\ell=1}^\infty \frac{t^{2\ell}}{2\ell} G_\ell(r)$. Thus, $t = \tanh(2s)$ is weighting how relevant each approximation is. For small squeezing, most of the weight is concentrated on low-order approximations. The lowest order approximation is $G_1(r) = r(1-r)$ resulting in a parabolic shaped Page curve. When the squeezing is very large, more and more weight is placed on high-order approximations so that the Page curve begins to resemble the triangle $G_\infty(r) = m(r)$. We see a manifestation of this interpretation as
\begin{align}
    \lim_{s \to 0} &\lim_{n \to \infty} \Expval_{U \in \U(n)} \frac{1}{s^2 n} S_2(U) = 2r(1-r), \\
    \lim_{s\to\infty} &\lim_{n\to\infty} \Expval_{U \in \U(n)} \frac{1}{s n} S_2(U) = 2 \min(r,1-r),
\end{align}
where the latter comes from the full expression in \cref{thm:page-curve}.
Meanwhile, the maximal R\'enyi-2 entropy is $\max_U \frac{1}{n}S_2(U) = m(r)\log\cosh(2s)$ from \cref{eq:max-s2}. As stated, near the endpoints $r=0$ and $r=1$, $G_\ell(r)$ is a very good approximation to $m(r)$. Thus, regardless of the squeezing strength, when the subsystem size $k = rn$ is small (or when its complement is small), the average entanglement is very close to maximal.

\subsection{Maximum value --- Proof of \texorpdfstring{\cref{cor:special-page-curve}}{\autoref*{cor:special-page-curve}}}
\label{ap:special-page}

In \cref{ap:page}, we derived the exact formula for the R\'enyi-2 Page curve as an infinite series. Here we will show that the series can be completely simplified when $r=1/2$. Bailey's theorem says that
\cite{petkovsek1996,bailey_generalized_1964,zudilin_hypergeometric_2019,slater_generalized_1966,wolfram_hypergeometric2f1}
\begin{equation}
    \, _2F_1(a,1-a;c;1/2) =  \frac{\Gamma(c/2)\Gamma((1+c)/2)}{\Gamma((a+c)/2)\Gamma((1+c-a)/2)}.
\end{equation}
Plugging this into the Page curve in \cref{thm:page-curve} at $r=1/2$ and simplifying with the duplication formula $\Gamma(\ell+1/2) = \frac{ \sqrt{\pi}(2\ell-1)!}{(\ell-1)! 2^{2\ell - 1}}$ \cite{abramowitz_handbook_2013},
\begin{equation}
    \lim_{n\to\infty}\Expval_{U\in\U(n)} \frac{1}{n}S_2(U)\\
    = \frac{1}{2}\log\cosh(2s)- \frac{1}{4}\sum_{\ell=1}^\infty \parentheses{\tanh^2(2s) / 4}^\ell \frac{1}{\ell} \binom{2\ell}{\ell}.
\end{equation}
We find that the second term is $\frac{1}{4}\int_0^{t/4} \frac{f(x)}{x}\dd{x}$, where $t=\tanh^2(2s)$ and $f(x)$ is the generating function of the central binomial coefficients $f(x)= \sum_{\ell=1}^\infty \binom{2\ell}{\ell}x^\ell$. Via the generalized binomial theorem, one finds the generating function evaluates to $f(x) = (1-4x)^{-1/2}-1$ \cite[A000984]{oeis}. Performing the integral, we find the second term to be $\frac{1}{4} \log \pargs{\frac{4-4 \sqrt{1-t}}{t+t\sqrt{1-t}} }$, which simplifies to $\frac{1}{2}\log \cosh(2s)-\log \cosh s$. Subtracting the second term from the first yields $\log \cosh s$ as desired.

Finally, using \cref{eq:max-s2} at $r=1/2$, we find the Page correction to be $\frac{1}{2}\log\cosh(2s) - \log \cosh s$, which simplifies to $\frac{1}{2}\log(1+\tanh^2 s)$.

\subsection{Constant term --- Proof of \texorpdfstring{\cref{prop:const}}{\autoref*{prop:const}}}
\label{ap:const}

In \cref{ap:page}, we found that when all the initial squeezers are equal to $s$, $\Expval_U S_2(U) = n\alpha(s,r) -\lambda(s,r) + \littleo{1}$, and we explicitly computed $\alpha(s,r)$. In this section, we determine $\lambda(s,r)$ up to a set of constants and conjecture an explicit value of the constants.

In \cref{ap:page}, we found that, asymptotically in $n$, $\Expval_U \Tr W^\ell = n f_\ell(r) + g_\ell(r) + \littleo{1}$, where $f_\ell$ and $g_\ell$ are functions of $r$. From \cref{eq:renyi-2}, it follows that $\lambda(s,r) = \sum_{\ell=1}^\infty \frac{1}{2\ell} \tanh^{2\ell}(2s) g_\ell(r)$.
Furthermore, in \cref{lem:form-of-w}, we found that $f_\ell$ takes the form $f_\ell(r) = \sum_{d=\ell+1}^{2\ell} \alpha_d^{(\ell)} r^d$. Indeed, the proof of \cref{lem:form-of-w} applies almost identically to $g_\ell(r)$, except that because we are interested now in the constant term in $\Expval_U \Tr W^\ell$ instead of the term linear in $n$, $g_\ell$ takes the form $g_\ell(r) = \sum_{d=\ell}^{2\ell} \beta_d^{(\ell)} r^d$. The extra term $\propto r^\ell$ in the polynomial $g_\ell$ does not occur in $f_\ell$.

The R\'enyi-2 entropy must be symmetric under $r \mapsto 1-r$ at every order in $s$ and $n$, since the full state on the $n$ modes is pure. Therefore, $\lambda(s,r)$ must be symmetric under $r \mapsto 1-r$ at every order in $\tanh(2s)$, meaning that $g_\ell(r) = g_\ell(1-r)$. The following lemma will therefore be useful here, and we will also find use of it in \cref{ap:variance}.

\medskip

\begin{lemma}
\label{lem:symmetric-polynomial}
Fix a polynomial $g_\ell(r) = \sum_{d=\ell}^{2\ell}\beta_d^{(\ell)} r^d$. If $g_\ell(r) = g_\ell(1-r)$, then $g_\ell(r) = \beta_{2\ell}^{(\ell)} (-1)^\ell \parentheses{r(1-r)}^\ell$.
\end{lemma}
\begin{proof}
Simplifying $g_l(r) - g_l(1-r) = 0$ with the binomial theorem, we find
\begin{align}
    0
    &= \sum_{d=\ell}^{2\ell} \beta_d^{(\ell)} \parentheses{r^d - (1-r)^d}\\
    &= -\sum_{d=0}^{\ell-1}(-r)^d \sum_{j=\ell}^{2\ell} \beta_j^{(\ell)} \binom{j}{d} + \sum_{d=\ell}^{2\ell} r^d \parentheses{\beta_d^{(\ell)} - (-1)^d \sum_{j=d}^{2\ell} \beta_j^{(\ell)} \binom{j}{d}}.
\end{align}
Equating all degrees of $r$, we find
\begin{enumerate}
    \item For $0 \leq d \leq \ell-1$, $0 = \sum_{j=\ell}^{2\ell} \beta_j^{(\ell)} \binom{j}{d}$;
    \item For $\ell \leq d \leq 2\ell$, $\beta_d^{(\ell)} = (-1)^d \sum_{j=d}^{2\ell} \beta_j^{(\ell)} \binom{j}{d}$.
\end{enumerate}
Condition 1 is a system of $\ell$ linearly independent equations. To verify this, one must show that $\det D \neq 0$, where $D$ is the $(\ell+1) \times (\ell+1)$ matrix with entries $D_{ij} = \binom{i+\ell-1}{j-1}$ for $1\leq i \leq \ell+1$ and $1 \leq j \leq \ell$, and entries $D_{i,\ell+1} = \delta_{i,\ell+1}$. In other words, the rightmost column is all zeros except for the entry on the diagonal. Inserting this rightmost column is equivalent to adding one new equation to the linear system, where this new equation simply fixes the value of one of the variables. Since the rightmost column is all zeros except for the $(\ell+1, \ell+1)$ entry, we use Laplace's expansion to find that $\det D = \det D'$, where $D'$ is the $\ell \times \ell$ upper left submatrix of $D$. One can prove that $\det D' \neq 0$ in a similar way as shown in Ref.~\cite{delanoy_determinant_2017}. Alternatively, one can use Corollary~11 from Ref.~\cite{grinberg_hyperfactorial}, which shows that the matrix with entries $\binom{a_i}{j-1}$ has zero determinant if and only if there are indices $i \neq j$ such that $a_i = a_j$. For the matrix $D'$, $a_i = \ell+i-1$, and therefore the determinant is nonzero.

In summary, we have a system of $\ell$ linearly independent equations for $\ell+1$ variables. We can therefore uniquely express the solution by fixing one of the variables. Suppose we know the value of $\beta_{2\ell}^{(\ell)}$. Then one can verify that conditions 1 and 2 are satisfied by $\beta_{\ell+i}^{(\ell)} = (-1)^{\ell+i} \binom{\ell}{i} \beta_{2\ell}^{(\ell)}$,
for $0 \leq i \leq \ell$. With the binomial theorem, this simply becomes $g_\ell(r) = \beta_{2\ell}^{(\ell)}(-1)^\ell\parentheses{r(1-r)}^\ell$.

When verifying that the two conditions are satisfied by $\beta_{\ell+i}^{(\ell)} = (-1)^{\ell+i} \binom{\ell}{i} \beta_{2\ell}^{(\ell)}$, one finds that the right hand side of both conditions can be simplified in terms of the hypergeometric function via \cref{eq:hypergeometric-function}. For example, condition 2, written as $(-1)^d = \sum_{j=d}^{2\ell} \binom{j}{d} (\beta_j^{(\ell)}/\beta_d^{(\ell)})$, reduces to $(-1)^d = \, _2F_1(d-2 \ell ,d+1;d-\ell +1;1)$. As in the proof of \cref{lem:proof-of-alpha}, we find that
\begin{align}
    \, _2F_1(d-2 \ell ,d+1;d-\ell +1;1)
    &= \frac{(-\ell)_{2\ell-d}}{(d-\ell+1)_{2\ell-d}}\\
    &= \frac{\Gamma(d-\ell+1)\Gamma(\ell-d)}{\Gamma(-\ell) \Gamma(\ell+1)}\\
    &= \frac{\Gamma(d-\ell+1)\Gamma(-\ell)(\ell-d-1)(\ell-d-2)\dots (-\ell)}{\Gamma(-\ell) \Gamma(\ell+1)}\\
    &= (-1)^{1+(\ell-d+1)-(-\ell)} \frac{\Gamma(d-\ell+1) (d+1-\ell)(d+2-\ell)\dots (\ell)}{\Gamma(\ell+1)}\\
    &= (-1)^d \frac{\Gamma(d-\ell+1) \ell!}{\Gamma(\ell+1) (d-\ell)!}\\
    &= (-1)^d,
\end{align}
which proves condition 2. The proof of condition 1 is similar.
\end{proof}

\medskip

For convenience, we define $a^{(\ell)} \coloneqq \beta_{2\ell}^{(\ell)}$. It follows from \cref{lem:symmetric-polynomial} that 
\begin{equation}
    \label{eq:const-from-a}
    \lambda(s,r) = \sum_{\ell=1}^\infty \frac{(-1)^\ell}{2\ell} a^{(\ell)}\tanh^{2\ell}(2s) \parentheses{r(1-r)}^\ell.
\end{equation}
Note that if $a^{(\ell)} = (-1)^{\ell}4^{\ell -1}$, then the Taylor series of $\log$ implies $\lambda(s,r) = -\frac{1}{8}\log\pargs{1-4r(1-r)\tanh^2(2s)}$, which is precisely \cref{prop:const}. Hence, the only thing left to prove is the following lemma.

\medskip

\begin{lemma}
    \label{lem:const}
    For all $\ell \in \bbN$, $a^{(\ell)} = (-1)^{\ell}4^{\ell -1}$.
\end{lemma}

\medskip

To make progress on \cref{lem:const}, we write the formula for $a^{(\ell)}$. Recall that $\beta_{2\ell}^{(\ell)}$ is the coefficient in front of the $r^{2\ell}$ term in $g_\ell(r)$. Looking at \cref{eq:weingarten-trw}, we see that the only way to get an $r^{2\ell}$ term (i.e.~a $k^{2\ell}$ term) is if $\sigma$ is the permutation $\sigma(1)=2\ell$, $\sigma(2) = 1$, $\sigma(3)=2$, $\dots$, $\sigma(2\ell) = 2\ell-1$. Therefore, we are interested in the constant term (the term asymptotically independent of $n$) in the following term of \cref{eq:weingarten-trw}:
\begin{equation}
\begin{aligned}
    &n^{2\ell}\sum_{j_1,\dots,j_{2\ell}=1}^n \sum_{j_1',\dots,j_{2\ell}'=1}^n \sum_{\tau \in S_{2\ell}} \text{Wg}(\sigma\tau^{-1},n)\\
    &\qquad \times \delta_{j_1,j_2}\delta_{j_1',j_2'}\dots \delta_{j_{2\ell-1},j_{2\ell}}\delta_{j_{2\ell-1}',j_{2\ell}'}\\
    &\qquad \times \delta_{j_1,j_{\tau(1)}'} \dots \delta_{j_{2\ell},j_{\tau(2\ell)}'},
\end{aligned}
\end{equation}
which simplifies to
\begin{align}
    \begin{split}
        &=n^{2\ell}\sum_{j_1,\dots,j_{2\ell}=1}^n \sum_{j_1',\dots,j_{2\ell}'=1}^n \sum_{\tau \in S_{2\ell}} \text{Wg}(\tau^{-1},n)\\
        &\qquad \times \delta_{j_1,j_2}\delta_{j_1',j_2'}\dots \delta_{j_{2\ell-1},j_{2\ell}}\delta_{j_{2\ell-1}',j_{2\ell}'}\\
        &\qquad \times \delta_{j_1,j_{\tau(2\ell)}'}\delta_{j_2,j_{\tau(1)}'}\delta_{j_3,j_{\tau(2)}'} \dots \delta_{j_{2\ell},j_{\tau(2\ell-1)}'}
    \end{split}\\
    \begin{split}
        &=n^{2\ell}\sum_{j_2,j_4,\dots,j_{2\ell}=1}^n \sum_{j_1',\dots,j_{2\ell}'=1}^n \sum_{\tau \in S_{2\ell}} \text{Wg}(\tau,n)\\
        &\qquad \times \delta_{j_1',j_2'}\dots \delta_{j_{2\ell-1}',j_{2\ell}'}\\
        &\qquad \times \delta_{j_2,j_{\tau(2\ell)}'}\delta_{j_2,j_{\tau(1)}'}\delta_{j_4,j_{\tau(2)}'}\delta_{j_4,j_{\tau(3)}'} \dots \delta_{j_{2\ell},j_{\tau(2\ell-2)}'}\delta_{j_{2\ell},j_{\tau(2\ell-1)}'}
    \end{split}\\
    \begin{split}
        &=n^{2\ell}\sum_{j_1',\dots,j_{2\ell}'=1}^n \sum_{\tau \in S_{2\ell}} \text{Wg}(\tau,n)\\
        &\qquad \times \delta_{j_1',j_2'}\dots \delta_{j_{2\ell-1}',j_{2\ell}'}\\
        &\qquad \times \delta_{j_{\tau(2\ell)}',j_{\tau(1)}'} \delta_{j_{\tau(2)}',j_{\tau(3)}'} \dots \delta_{j_{\tau(2\ell-2)}',j_{\tau(2\ell-1)}'}.
    \end{split}
\end{align}
Note we used that $\mathrm{Wg}(\tau,n) = \mathrm{Wg}(\tau^{-1}, n)$ from \cref{eq:asymptotic-weingarten}. Let $\#(\tau)$ be the number of cycles in the disjoint cycle decomposition of $\tau$ and let $\{c_i^{(\tau)} \mid i\in \set{1,\dots, \#(\tau)}\}$ be the cycle decomposition. Then, from \cref{eq:asymptotic-weingarten}, asymptotically,
\begin{equation}
    \text{Wg}(\tau,n) = \frac{1}{n^{2\ell+\abs{\tau}}} \prod_{i=1}^{\#(\tau)} (-1)^{\lvert c_i^{(\tau)}\rvert -1} C_{\lvert c_i^{(\tau)}\rvert - 1},
\end{equation}
where recall $C_m$ is the $m^{\rm th}$ Catalan number. Therefore, $a^{(\ell)}$ is the constant term independent of $n$ in
\begin{equation}
\begin{aligned}
    &\sum_{j_1,\dots,j_{2\ell}=1}^n \sum_{\tau \in S_{2\ell}} \frac{1}{n^{\abs{\tau}}} \prod_{i=1}^{\#(\tau)} (-1)^{\lvert c_i^{(\tau)}\rvert-1} C_{\lvert c_i^{(\tau)}\rvert - 1} \\
    &\qquad \times \delta_{j_1,j_2}\dots \delta_{j_{2\ell-1},j_{2\ell}}\\
    &\qquad \times \delta_{j_{\tau(2\ell)},j_{\tau(1)}} \delta_{j_{\tau(2)},j_{\tau(3)}} \dots \delta_{j_{\tau(2\ell-2)},j_{\tau(2\ell-1)}}.
\end{aligned}
\end{equation}

\setlength{\tabcolsep}{12pt}

\begin{table}
    \centering
    \begin{tabular}{lll}
        \toprule
         $\ell$ & $a^{(\ell)}$ & \\
         \midrule
         $1$ & $-C_1$ & $= -1$ \\
         $2$ & $-4 C_0^2 C_1+4 C_0 C_2$ & $=4$ \\
         $3$ & $-9 C_0^4 C_1+15 C_0^2 C_1^2-C_1^3+18 C_0^3 C_2-6 C_0 C_1 C_2-9 C_0^2 C_3$ & $=-16$\\
         $4$ & $-16 C_0^6 C_1+80 C_0^4 C_1^2-40 C_0^2 C_1^3+48 C_0^5 C_2-112 C_0^3 C_1 C_2$ & \\
         & $~~+16 C_0 C_1^2 C_2+8 C_0^2 C_2^2-48 C_0^4 C_3+24 C_0^2 C_1 C_3+16 C_0^3 C_4$ & $=64$ \\
         $5$ & $-25 C_0^8 C_1+250 C_0^6 C_1^2-380 C_0^4 C_1^3+80 C_0^2 C_1^4-C_1^5+100 C_0^7 C_2$ & \\
         & $~~ -600 C_0^5 C_1 C_2+440 C_0^3 C_1^2 C_2-20 C_0 C_1^3 C_2 + 130 C_0^4 C_2^2$ & \\
         & $~~ -50 C_0^2 C_1 C_2^2 -150 C_0^6 C_3+320 C_0^4 C_1 C_3-60 C_0^2 C_1^2 C_3-40 C_0^3 C_2 C_3$ & \\
         & $~~ +100 C_0^5 C_4-60 C_0^3 C_1 C_4-25 C_0^4 C_5$ & $=-256$\\
         \bottomrule
    \end{tabular}
    \caption{A table showing the first five values of $a^{(\ell)}$ from \cref{eq:a-ell}, which matches \cref{lem:const}.}
    \label{tab:const}
\end{table}

Define $\xi\colon S_{2\ell} \to \bbN$ such that
\begin{align}
    n^{\xi(\tau)}
    &= \sum_{j_1,\dots,j_{2\ell}=1}^n \delta_{j_1,j_2}\delta_{j_3,j_4}\dots \delta_{j_{2\ell-1},j_{2\ell}} \times \delta_{j_{\tau(2\ell)},j_{\tau(1)}} \delta_{j_{\tau(2)},j_{\tau(3)}} \dots \delta_{j_{\tau(2\ell-2)},j_{\tau(2\ell-1)}}\\
    &= \sum_{j_1,\dots,j_\ell=1}^n \delta_{j_{\ceil{\tau(2\ell)/2}},j_{\ceil{\tau(1)/2}}} \delta_{j_{\ceil{\tau(2)/2}},j_{\ceil{\tau(3)/2}}} \dots \delta_{j_{\ceil{\tau(2\ell-2)/2}},j_{\ceil{\tau(2\ell-1)/2}}}.
\end{align}
Note that the definition of $\xi$ is independent of the value of $n$, and so we can equivalently define $\xi$ as
\begin{equation}
    \xi(\tau) = \log_2 \sum_{j_1,\dots,j_\ell=1}^2 \delta_{j_{\ceil{\tau(2\ell)/2}},j_{\ceil{\tau(1)/2}}} \delta_{j_{\ceil{\tau(2)/2}},j_{\ceil{\tau(3)/2}}} \dots \delta_{j_{\ceil{\tau(2\ell-2)/2}},j_{\ceil{\tau(2\ell-1)/2}}}.
\end{equation}
To get the constant term, we need $\frac{n^{\xi(\tau)}}{n^{\abs{\tau}}} = 1$. Therefore,
\begin{equation}
    a^{(\ell)}
    = \sum_{\substack{\tau\in S_{2\ell} \text{ s.t.}\\\xi(\tau) = \abs{\tau}}} \prod_{i=1}^{\#(\tau)} (-1)^{\lvert c_i^{(\tau)}\rvert-1} C_{\lvert c_i^{(\tau)}\rvert - 1}.
\end{equation}
Finally, since the sum of the lengths of the cycles of a permutation $\tau \in S_{2\ell}$ is always $2\ell$, we find
\begin{equation}
    \label{eq:a-ell}
    a^{(\ell)}
    = \sum_{\substack{\tau\in S_{2\ell} \text{ s.t.}\\\xi(\tau) = \abs{\tau}}} (-1)^{\#(\tau)} \prod_{i=1}^{\#(\tau)} C_{\lvert c_i^{(\tau)}\rvert - 1}.
\end{equation}
From this equation, we can exactly compute $a^{(\ell)}$ on a computer for small values of $\ell$. \cref{tab:const} shows the first five of these, which all match \cref{lem:const} saying that $a^{(\ell)} = (-1)^\ell 4^{\ell-1}$. Note that if we change the condition of $\xi(\tau) = \abs{\tau}$ in \cref{eq:a-ell} to $\xi(\tau) = \abs{\tau} + 1$, then this gives us the term that is linear in $n$ and hence is the equation for $\alpha_{2\ell}^{(\ell)}$ from \cref{lem:proof-of-alpha}, which is $\frac{(-1)^{\ell +1}}{2 \ell -1} \binom{2 \ell -1}{\ell -1}$.

To complete the proof of \cref{lem:const}, we need to evaluate \cref{eq:a-ell} for all $\ell \in \bbN$. This is done in Ref.~\cite{alekseyev} using objects called \textit{breakpoint graphs} that arise in the study of gene orders in bioinformatics \cite{Alexeev2017hultman}. Roughly, $\xi$ has an interpretation in terms of cycles of breakpoint graphs.

\section{Variance of the R\'enyi-2 entropy --- Proof of \texorpdfstring{\cref{thm:renyi2-variance}}{\autoref*{thm:renyi2-variance}}}
\label{ap:variance}

In this section, we shift our attention away from $\Expval_U S_2(U)$ and instead to $\variance_U S_2(U) = \Expval_U S_2(U)^2 - \parentheses{\Expval_U S_2(U)}^2$, and we will prove \cref{thm:renyi2-variance}. We are again interested in the case where all the initial squeezers are equal to $s$. Using \cref{eq:renyi-2}, this becomes
\begin{align}
    \variance_{U\in\U(n)}S_2(U)
    &= \sum_{\ell,\ell'=1}^\infty \frac{1}{(2\ell)(2\ell')}\tanh^{2\ell+2\ell'}(2s) \parentheses{\Expval_U (\Tr W^\ell)(\Tr W^{\ell'}) - \parentheses{\Expval_U \Tr W^\ell}\parentheses{\Expval_U \Tr W^{\ell'}} }\\
    &= \sum_{d=2}^\infty \tanh^{2d}(2s) \sum_{\ell=1}^{d-1} \frac{1}{4\ell(d-\ell)} \parentheses{\Expval_U (\Tr W^\ell)(\Tr W^{d-\ell}) - \parentheses{\Expval_U \Tr W^\ell}\parentheses{\Expval_U \Tr W^{d-\ell}} }.\label{eq:exact-variance}
\end{align}
As a direct consequence of \cref{lem:form-of-w}, we find that asymptotically
\begin{equation}
    \parentheses{\Expval_U \Tr W^\ell}\parentheses{\Expval_U \Tr W^{\ell'}} = n^2 p_{\ell,\ell'}(r) + n q_{\ell,\ell'}(r)+t_{\ell,\ell'}(r) + \littleo{1},
\end{equation}
where $p_{\ell,\ell'}$ is a polynomial of degrees $\ell+\ell'+2$ through $2\ell+2\ell'$ in $r$, $q_{\ell,\ell'}$ is a polynomial of degrees $\ell+\ell'+1$ through $2\ell+2\ell'$ in $r$, and $t_{\ell,\ell'}$ is a polynomial of degrees $\ell+\ell'$ through $2\ell+2\ell'$ in $r$. 

Furthermore, we find an analogous result for $\Expval_U (\Tr W^\ell)(\Tr W^{\ell'})$. Let $L \coloneqq 2\ell+2\ell'$. Using \cref{eq:weingarten-formula}, we find that
\begin{equation}
\label{eq:weingarten-trw-trw}
\begin{aligned}
    \Expval_{U \in \U(n)} (\Tr W^\ell)(\Tr W^{\ell'})
    &= \sum_{i_1,\dots i_{L}=1}^k \sum_{i_1',\dots i_{L}'=1}^k \sum_{j_1,\dots j_{L}=1}^n \sum_{j_1',\dots j_{L}'=1}^n \sum_{\sigma,\tau \in S_L} \mathrm{Wg}(\sigma\tau^{-1}, n)\\
    &\qquad \times \delta_{i_{2\ell}',i_1}\delta_{i_1',i_2}\delta_{i_2',i_3}\dots \delta_{i_{2\ell-1}',i_{2\ell}} \\
    &\qquad\times \delta_{i_{L}',i_{2\ell+1}}\delta_{i_{2\ell+1}',i_{2\ell+2}}\delta_{i_{2\ell+2}',i_{2\ell+3}}\dots \delta_{i_{L-1}',i_L}\\
    &\qquad \times \delta_{j_1,j_2}\delta_{j_1',j_2'}\dots \delta_{j_{2\ell-1},j_{2\ell}}\delta_{j_{2\ell-1}',j_{2\ell}'} \\
    &\qquad\times \delta_{j_{2\ell+1},j_{2\ell+2}}\delta_{j_{2\ell+1}',j_{2\ell+2}'}\dots \delta_{j_{L-1},j_{L}}\delta_{j_{L-1}',j_{L}'}\\
    &\qquad \times \delta_{i_1,i_{\sigma(1)}'} \dots \delta_{i_{L},i_{\sigma(L)}'}\\
    &\qquad \times \delta_{j_1,j_{\tau(1)}'} \dots \delta_{j_{L},j_{\tau(L)}'},
\end{aligned}
\end{equation}
and the asymptotic form of $\mathrm{Wg}$ function is given in \cref{eq:asymptotic-weingarten}. In a similar proof to \cref{lem:form-of-w} but with \cref{eq:weingarten-trw-trw} instead of \cref{eq:weingarten-trw}, we analogously find that asymptotically
\begin{equation}
    \Expval_{U \in \U(n)} (\Tr W^\ell)(\Tr W^{\ell'}) = n^2 P_{\ell,\ell'}(r) + n Q_{\ell,\ell'}(r) + T_{\ell,\ell'}(r) + \littleo{1},
\end{equation}
where $P_{\ell,\ell'}$ is a polynomial of degrees $\ell+\ell'+2$ through $2\ell+2\ell'$ in $r$, $Q_{\ell,\ell'}$ is a polynomial of degrees $\ell+\ell'+1$ through $2\ell+2\ell'$ in $r$, and $T_{\ell,\ell'}$ is a polynomial of degrees $\ell+\ell'$ through $2\ell+2\ell'$ in $r$. For completeness, we prove this result for the $P_{\ell,\ell'}$ term in the following lemma. The proofs for the $Q_{\ell,\ell'}$ and $T_{\ell,\ell'}$ terms follow from trivially tweaking the final part of the proof.

\medskip

\begin{lemma}
\label{lem:form-of-ww}
Fix positive integers $\ell$ and $\ell'$. There exist coefficients $\beta_d^{(\ell,\ell')}$ for $d \in \set{\ell+\ell'+2,\dots, 2\ell+2\ell'}$ such that
\begin{equation}
    P_{\ell,\ell'}(r) \coloneqq \lim_{n\to\infty}\frac{1}{n^2}\Expval_U (\Tr W^\ell)(\Tr W^{\ell'}) = \sum_{d=\ell+\ell'+2}^{2\ell+2\ell'} \beta_d^{(\ell,\ell')}r^d.
\end{equation}
\end{lemma}
\begin{proof}
Much of the details of this proof are the same as in the proof of \cref{lem:form-of-w}. We will use the asymptotic form of the Wg function, which is written in \cref{eq:asymptotic-weingarten}. 
The proof will proceed as follows. First, we will prove that $(\Tr W^\ell)(\Tr W^{\ell'})$ contains a term proportional to $n^2$ and no terms proportional to $n^a$ for any $a > 2$. Therefore, $P_{\ell,\ell'}(r)$ is indeed independent of $n$. Next, we will prove that $P_{\ell,\ell'}(r)$ has no terms $r^a$ for $a > L$ and $a\leq \ell+\ell'+1$. Throughout this proof, we interpret the delta functions in \cref{eq:weingarten-trw-trw} as constraints on the summations. Different permutations on the indices result in a different number of constraints and hence terms with different powers of $n$ and $k$.

Recall that $\Pi$ can only decrease the trace. Getting rid of the $\Pi$ in $W$ and using the cyclic nature of the trace, we see that $(\Tr W^\ell)(\Tr W^{\ell'}) \leq n^2$. Furthermore, consider \cref{eq:weingarten-trw-trw} with $\sigma=\tau$ defined by $\sigma(1)=2\ell$, $\sigma(2) = 1$, $\dots$, $\sigma(2\ell) = 2\ell-1$ and $\sigma(2\ell+1) = L$, $\sigma(2\ell+2) = 2\ell+1$, $\dots$, $\sigma(L) = L-1$. Then, $\sigma \tau^{-1}$ is the identity, and so $\mathrm{Wg}(\sigma\tau^{-1},n)$ contributes a factor of $n^{-L}$. With this $\sigma$, the sum over the $i$ and $i'$ yields a factor of $k^L$. Finally, the sum over $j$ yields
\begin{equation}
\delta_{j_1',j_2'}\delta_{j_3',j_4'}\dots \delta_{j_{L-1}',j_{L}'} \delta_{j_1',j_{2\ell}'} \delta_{j_2',j_1'}\dots \delta_{j_{2\ell}',j_{2\ell-1}'}\delta_{j_{2\ell+1}',j_{L}'} \delta_{j_{2\ell+2}',j_{2\ell+1}'}\dots \delta_{j_{L}',j_{L-1}'}. 
\end{equation}
Then summing over $j'$, we get two factors of $n$. Hence, the term with the specific permutation described above yields a term of the form $n^2 k^{L} n^{-L} = n^2 r^{L}$.

We have shown that there is a term proportional to $n^2$ and that there are no terms proportional to $n^a$ for $a > 2$. Since we are working asymptotically in $n$, we can therefore ignore all terms proportional to $n^a$ for every $a < 2$. This proves that $\lim_{n\to\infty} \frac{1}{n^2}(\Tr W^\ell)(\Tr W^{\ell'})$ is independent of $n$ and only depends on $r$, which justifies the definition of the function $P_{\ell,\ell'}(r)$. The only thing left to show is that $P_{\ell,\ell'}(r)$ has only terms $r^{\ell+\ell'+2}$ through $r^{2\ell+2\ell'} = r^L$. So we only need to show that there are no terms $r^a$ for $a > L$ and $a \leq \ell+\ell'+1$. We begin with the former.

To look at powers of $r$, it is sufficient to look at powers of $k$. We therefore restrict our attention to the sum over $i$ and $i'$ in \cref{eq:weingarten-trw-trw}. In order to get the highest power of $k$, we require the least constraints on $i$ and $i'$ (i.e.~the least distinct Kronecker deltas). We therefore require $\sigma$ to be the permutation so that
\begin{equation}
     \delta_{i_{2\ell}',i_1}\delta_{i_1',i_2}\delta_{i_2',i_3}\dots \delta_{i_{2\ell-1}',i_{2\ell}}  \delta_{i_{L}',i_{2\ell+1}}\delta_{i_{2\ell+1}',i_{2\ell+2}}\delta_{i_{2\ell+2}',i_{2\ell+3}}\dots \delta_{i_{L-1}',i_L} = \delta_{i_1,i_{\sigma(1)}'} \dots \delta_{i_{L},i_{\sigma(L)}'}.
\end{equation}
This permutation $\sigma$ is exactly the $\sigma$ described above that gave the term proportional to $n^2 r^L$. Hence, $L$ is the highest power of $r$ that can be achieved.

Next we need to show that $\ell+\ell'+2$ is the lowest power of $k$ that can be achieved. The sum over $j$ and $j'$ can give at most a factor of $n^{\ell+\ell'}$. This is because the first line of delta functions, $\delta_{j_1,j_2}\delta_{j_1',j_2'}\dots \delta_{j_{L-1},j_L}\delta_{j_{L-1}',j_L'}$, reduces the sum over the $L$ indices $j$ and the $L$ indices $j'$ down to just a sum over $\ell+\ell'$ indices $j$ and $\ell+\ell'$ indices $j'$. The second line of delta functions, $\delta_{j_1,j_{\tau(1)}'}\dots \delta_{j_L,j_{\tau(L)}'}$, cannot be made equivalent to the first line by any choice of $\tau$; in fact, the second line imposes all new constraints. Therefore, this line further reduces the sum over the $\ell+\ell'$ indices $j$ and the $\ell+\ell'$ indices $j'$ to just a sum over the $\ell+\ell'$ indices $j$ (or the $\ell+\ell'$ indices $j'$, but not both). Hence the highest power of $n$ that we can get from the summations over the $j$ and $j'$ is $n^{\ell+\ell'}$. Putting this together with the fact that asymptotically $\mathrm{Wg}(\pi, n)$ is at most $n^{-L} = n^{-2\ell-2\ell'}$, we find that any term coming from \cref{eq:weingarten-trw-trw} is at most $n^{-\ell-\ell'} \times (\text{dependence on }k)$. Therefore, any powers of $k$ that are less than $\ell+\ell'+2$ can be ignored; if the sum over $i$ and $i'$ yields a term that is $k^a$ for some $a \leq \ell+\ell'+1$, then that term will scale linearly or less with $n$. But from above, we already have terms that are proportional to $n^2$, and so terms that are linear or less in $n$ can be ignored.
\end{proof}

\medskip

Therefore, from \cref{eq:exact-variance}, we find that asymptotically
\begin{equation}
    \variance_{U\in\U(n)} S_2(U) = \sum_{d=2}^\infty \tanh^{2d}(2s) \parentheses{n^2 p_d(r) + n q_d(r) + t_d(r)} + \littleo{1},
\end{equation}
where $p_d$ is a polynomial of degrees $d+2$ through $2d$ in $r$, $q_d$ is a polynomial of degrees $d+1$ through $2d$ in $r$, and $t_d$ is a polynomial of degrees $d$ through $2d$ in $r$. The variance must be symmetric under $r\mapsto 1-r$ at every order in $s$ and $n$. Therefore each $p_d$, $q_d$, and $t_d$ must themselves be symmetric under $r \mapsto 1-r$. It then immediately follows from \cref{lem:symmetric-polynomial} that $t_d(r) = \omega^{(d)} \parentheses{r(1-r)}^d$ for some constants $\omega^{(d)}$. In the following lemma, we show that $p_d$ and $q_d$ must be the zero polynomial.

\medskip

\begin{lemma}
Let $d$ be a positive integer, and $f(r) = \sum_{j=d+1}^{2d} \gamma_j r^j$. If $f(r) = f(1-r)$, then $f(r) = 0$.
\end{lemma}
\begin{proof}
Expanding with the binomial theorem, we find
\begin{align}
    0
    &= f(r) - f(1-r)\\
    &= \sum_{j=d+1}^{2d} \gamma_j \parentheses{r^j - (1-r)^j}\\
    &= \sum_{j=d+1}^{2d} \gamma_j r^j - \sum_{j=d+1}^{2d} \gamma_j \sum_{i=0}^j (-r)^i \binom{j}{i}\\
    &= \sum_{j=d+1}^{2d} \gamma_j r^j -  \sum_{i=0}^{2d} (-r)^i \sum_{j=\max(i,d+1)}^{2d} \gamma_j  \binom{j}{i}\\
    &= -\sum_{j=0}^{d} (-r)^j \sum_{i=d+1}^{2d} \gamma_i \binom{i}{j} + \sum_{j=d+1}^{2d} r^j \parentheses{\gamma_j - (-1)^j \sum_{i=j}^{2d} \gamma_i \binom{i}{j}}.
\end{align}
This must be true to all orders in $r$, and hence we equate each degree of $r$ to zero. This gives us the conditions:
\begin{enumerate}
    \item For $0 \leq j \leq d$, $0 =\sum_{i=d+1}^{2d} \gamma_i \binom{i}{j}$;
    \item For $d+1 \leq j \leq 2d$, $\gamma_j = (-1)^j \sum_{i=j}^{2d} \gamma_i \binom{i}{j}$.
\end{enumerate}
Choosing any $d$ of the equations from the first condition gives a linear system of $d$ linearly independent equations. To verify this, one must show that $\det C \neq 0$, where $C$ is the $d \times d$ matrix with entries $C_{ij} = \binom{d+i}{j}$. This was shown in the proof of \cref{lem:proof-of-alpha}.

Therefore, we have $d$ linearly independent equations for $d$ variables. Hence, if there is a solution, then there is one \textit{unique} solution. We easily see that $\gamma_j =0$ is a solution, and therefore it is the solution. This gives $f(r) = 0$, completing the proof.
\end{proof}

\medskip

Since $p_d$ and $q_d$ are zero, we have therefore found that
\begin{equation}
    \lim_{n\to\infty}\variance_{U\in\U(n)} S_2(U) = \sum_{d=2}^\infty \omega^{(d)}\tanh^{2d}(2s) \parentheses{r(1-r)}^d,
\end{equation}
proving the first part of \cref{thm:renyi2-variance}. Equating this equation to \cref{eq:exact-variance}, we find that
\begin{align}
    \omega^{(d)}
    &= \lim_{n\to\infty} \parentheses{r(1-r)}^{-d}\sum_{\ell=1}^{d-1} \frac{1}{4\ell(d-\ell)} \parentheses{\Expval_U (\Tr W^\ell)(\Tr W^{d-\ell}) - \parentheses{\Expval_U \Tr W^\ell}\parentheses{\Expval_U \Tr W^{d-\ell}} }\\
    &= \lim_{n\to\infty}\parentheses{r(1-r)}^{-d}\sum_{\ell=1}^{d-1}\frac{1}{4\ell(d-\ell)}\covariance_{U\in\U(n)}\pargs{\Tr W^\ell, \Tr W^{d-\ell}},
\end{align}
where $\covariance$ is the covariance, and we know that $\omega^{(d)}$ is independent of $r$ and $n$. From our expressions for $\Tr W^\ell$ in terms of the Weingarten calculus, it follows that $\omega^{(d)} \in \bbQ$. Recall that $W$ has two factors of $U$ and two factors of $\bar U$. Hence, we can exactly compute $\omega^{(2)}$ by integrating over fourth moments of the Haar measure on the unitary group. To do this, we use the Mathematica package \texttt{RTNI} that symbolically computes expressions over the Haar measure \cite{fukuda_rtnisymbolic_2019}. This Mathematica package precomputes the symbolic expressions for $\mathrm{Wg}(\sigma, n)$ for $\sigma \in S_4$. From this, $\omega^{(2)}$ is a sum over powers of $\Tr \Pi = k$ with coefficients depending on $n$. One can then simplify this expression and take the limit $n\to\infty$ to find that $\omega^{(2)} = 1/2$. Our Mathematica code for this calculation is provided on GitHub \cite{joseph_t_iosue_glo_2022}.

\bibliographystyle{quantum}
\bibliography{references}

\end{document}